\documentclass[pra,twocolumn,a4paper,nofootinbib, superscriptaddress]{revtex4-1}

\usepackage{hyperref}
\usepackage{ulem}
\usepackage{graphicx}
\usepackage{bm}
\usepackage{amsmath}
\usepackage{amssymb}
\usepackage{amsthm}
\usepackage{color}

\theoremstyle{plain}% default
\newtheorem{thm}{Theorem}[section]
\newtheorem{lem}[thm]{Lemma}

\theoremstyle{definition}

\theoremstyle{remark}

\definecolor{nblue}{rgb}{0.2,0.2,0.7}
\definecolor{ngreen}{rgb}{0.1,0.5,0.1}%272
\definecolor{nred}{rgb}{0.8,0.2,0.2}%711
\definecolor{nblack}{rgb}{0,0,0}

\newcommand{\unit}{I}
\newcommand{\tr}[1]{\text{Tr}\left(#1\right)}
\newcommand{\K}{\mathcal{K}}
\renewcommand{\O}{\mathcal{O}}
\newcommand{\Z}{\mathbb{Z}}
\newcommand{\SN}{\mathcal{S}^N}
\newcommand{\SCHSH}[1]{\mathcal{S}^{\mbox{\tiny CHSH}}_{#1}}
\newcommand{\I}{\mathcal{I}}
\newcommand{\ket}[1]{|#1\rangle}
\newcommand{\bra}[1]{\langle #1 |}
\newcommand{\pNL}{p_\text{NL}}
\newcommand{\pM}{p_\text{MABK}}
\newcommand{\pW}{p_\text{WWZB}}
\newcommand{\GHZ}{\ket{\Psi_N}}
\newcommand{\GHZrho}{\ket{\Psi_N}\!\bra{\Psi_N}}
\newcommand{\proj}[1]{\ket{#1}\!\bra{#1}}

\newcommand{\N}{\mathcal{N}}
\newcommand{\nn}{{n_0}}

\begin{document}

\title{Generating nonclassical correlations without fully aligning measurements}

\author{Joel J. \surname{Wallman}}
\affiliation{School of Physics, The University of Sydney,
Sydney, New South Wales 2006, Australia}
\author{Yeong-Cherng \surname{Liang}}
\affiliation{School of Physics, The University of Sydney,
Sydney, New South Wales 2006, Australia}
\affiliation{Group of Applied Physics, University of Geneva, CH-1211 Geneva 4, Switzerland}
\author{Stephen D. \surname{Bartlett}}
\affiliation{School of Physics, The University of Sydney,
Sydney, New South Wales 2006, Australia}

\date{\today}

\begin{abstract}
We investigate the scenario where spatially separated parties perform measurements in randomly chosen bases on an $N$-partite Greenberger-Horne-Zeilinger state. We show that without any alignment of the measurements, the observers will obtain correlations that violate a Bell inequality with a probability that rapidly approaches 1 as $N$ increases and that this probability is robust against noise. We also prove that restricting these randomly chosen measurements to a plane perpendicular to a common direction will \textit{always} generate correlations that violate some Bell inequality. Specifically, if each observer chooses their two measurements to be locally orthogonal, then the $N$ observers will violate one of two Bell inequalities by an amount that increases exponentially with $N$. These results are also robust against noise and perturbations of each observer's reference direction from the common direction.
\end{abstract}


\maketitle

\section{Introduction}

Entangled quantum states can yield correlations between spatially separated measurements that are inconsistent with any locally causal theory~\cite{Bell1964, bell2004}. These nonlocal (or nonclassical) correlations are a resource~\cite{barret2005a, brunner2005} for a range of information processing tasks such as quantum key distribution~\cite{ekert1991, barret2005b, acin2006, masanes2009}, teleportation~\cite{bennet1993,horodeckis1996}, certification and expansion of randomness~\cite{BIV:Randomness,R.Colbeck:1011.4474} and the reduction~\cite{brukner2004} of communication complexity~\cite{CC}.

These nonlocal correlations are correlations of measurement outcomes. As such, they are not solely a consequence of
entanglement but also depend upon the choice of measurements.This point was emphasised by Bell in his seminal work~\cite{Bell1964}, where he showed that the perfect correlations exhibited by the spin-singlet state do admit a simple locally causal model (see also Ref.~\cite{BIVvsEnt} and references therein).

The demonstration of nonlocal correlations typically employs carefully chosen measurements whose implementation requires the spatially separated observers to share a complete reference frame~\cite{aspect2004,bartlett2007a}. To circumvent this requirement, observers who do not initially share a complete reference frame could share a particular state that is invariant under arbitrary rotations of the local reference frames~\cite{cabello2003}; however, the state preparation involved is relatively complicated. Alternatively, they could use correlated quantum systems to establish a shared reference frame which can then be used to align measurements~\cite{rff}; however, this approach is resource-intensive as it requires coherently exchanging many entangled quantum systems.

Recently, it has been shown that such methods are not required to demonstrate violations of a Bell inequality~\cite{ycliang2010}. In particular, for $N$ spatially separated observers that share a Greenberger-Horne-Zeilinger (GHZ) state, it was found that most choices of measurements lead to nonlocal correlations between measurement outcomes~\cite{ycliang2010}. Therefore distant observers can randomly choose measurements that violate some Bell inequality with a probability that approaches 1 as $N$ increases. However, the successful detection of nonlocal correlations in this scenario requires checking the measurement statistics against a set of Bell inequalities that grows exponentially in $N$.

In this paper, we show that for $N$ observers who share a GHZ state and are able to perfectly align a single measurement direction (but do not share a full reference frame), any choice of measurements satisfying a local constraint will generate nonlocal correlations. Furthermore, in contrast to the results of Ref.~\cite{ycliang2010}, verifying that these correlations are nonlocal only involves testing the measurement statistics against two Bell inequalities, thereby simplifying the verification process. Moreover, we show that as $N$ increases, the amount by which the observers violate one of the two Bell inequalities increases exponentially and, in the worst-case scenario, is a constant factor below the maximum violation permitted by quantum mechanics.

We also investigate the robustness of the above-mentioned results and those presented in Ref.~\cite{ycliang2010} in the presence of uncorrelated local noise. We will demonstrate that the ability of observers to obtain correlations that violate some Bell inequality is robust against depolarizing or dephasing noise whether or not they share a direction. Finally, we show that even if the observers cannot perfectly align a measurement direction and so have randomly perturbed approximations to the common direction, they can still always obtain measurement statistics that violate one of two Bell inequalities.

This paper is structured as follows. We begin in Sec.~\ref{sec:preliminaries} by illustrating these results in the simplest case, namely, when there are two observers who share a Bell state. We then outline the generalization to $N$ parties, discuss the relevant Bell inequalities and the methods of sampling random measurements that are used in this paper. In Sec.~\ref{sec:Noiseless}, we consider the ideal case in which the $N$ parties share a state without any noise and can also share a reference direction perfectly. In Sec.~\ref{Sec:Decoherence}, we relax the first assumption and show that the probability of violating a Bell inequality is robust against uncorrelated depolarizing or dephasing noise. In Sec.~\ref{Sec:Perturbed} we also show that the probability of violating a Bell inequality is robust against random perturbations in the shared direction. In Sec.~\ref{sec:conclusion} we discuss the implications of these results and offer some concluding remarks.

\section{Preliminaries}\label{sec:preliminaries}

In this section, we outline the simplest example, namely, when two spatially separated parties each perform two binary-outcome measurements on the Bell state $\ket{\Phi^+}$. For this case, we define the probability of violating a Bell inequality and review the results presented in Ref.~\cite{ycliang2010} for randomly chosen measurements without any shared reference frame. We then show that if both observers perform locally orthogonal random measurements in the $xy$-plane of the Bloch sphere, they will always obtain correlations that violate a Clauser-Horne-Shimony-Holt (CHSH) Bell inequality~\cite{bell2004, CHSH}. We then outline the general scenario for $N$ parties and discuss the relevant Bell inequalities. We conclude this section by explaining the various samplings of measurement bases employed in this paper.

\subsection{A two-party example}\label{sec:2party_example}

For the two-party case, a verifier prepares many copies of the maximally entangled Bell state
\begin{equation}
\ket{\Phi^{+}} = \frac{1}{\sqrt{2}}\left(\ket{0}\ket{0} + \ket{1}\ket{1}\right)	\,,\label{eq:bell_state}
\end{equation}
where $\ket{0}$, $\ket{1}$ are the computational basis states, and distributes one qubit to each of two observers. Both observers choose two measurement bases. For each copy of the Bell state, the observers each randomly choose and perform one of their two measurements on their qubit. The observers then send the verifier a list of the measurement choice (a binary digit $s_k$) and the corresponding outcome $o^k_{s_k}=\pm1$ for each qubit $k$. The verifier uses the lists to determine if the measurement outcomes are inconsistent with a locally causal theory. The verifier does this by calculating the probabilities $p(o^1_{s_1}=o^2_{s_2})$ (as relative frequencies) that the outcomes satisfy $o^1_{s_1} = o^2_{s_2}$ given a specific choice of $s_1$ and $s_2$. They then determine the correlation functions
\begin{align}
E\left(s_1,s_2\right) &= p(o^1_{s_1}=o^2_{s_2}) - p(o^1_{s_1}=-o^2_{s_2})	\nonumber\\
&= 2p(o^1_{s_1}=o^2_{s_2}) - 1	\,,
\end{align}
and seek to determine if the correlation functions are consistent with a locally causal model. For two parties, the correlation functions are inconsistent with a locally causal theory if they violate the standard CHSH~\cite{bell2004, CHSH} Bell inequality
\begin{equation}
\SCHSH{1}=|E(0,0)+E(0,1)+E(1,0)-E(1,1)|\leq 2\,.\label{eq:CHSH}
\end{equation}
However, the choice of the labels for the measurements (i.e., which measurement is labelled by $s_k = 0$) is arbitrary, as is the labeling of the measurement outcomes $o^k_{s_k}$. Therefore the correlation functions are also inconsistent with any locally causal theory if they violate inequality~\eqref{eq:CHSH} after any relabeling of the $s_k$ and/or $o^k_{s_k}$ and/or the label $k$. We will follow the terminologies of Refs.~\cite{masanes2003,collins2004} and refer to two inequalities that can be obtained from one another through such relabeling as being \textit{equivalent}.

There are four equivalent inequalities that can be obtained from Eq.~\eqref{eq:CHSH} by mapping $(s_1,s_2)$ to $(s_1,s_2)$, $(s_1\oplus1,s_2)$, $(s_1,s_2\oplus1)$ or $(s_1\oplus1,s_2\oplus1)$, namely,
\begin{align}\label{eq:CHSH4}
\SCHSH{1}=|E(0,0)+E(0,1)+E(1,0)-E(1,1)|	&\leq 2	\,,\nonumber\\
\SCHSH{2}=|E(0,0)-E(0,1)+E(1,0)+E(1,1)|	&\leq 2	\,,\nonumber\\
\SCHSH{3}=|E(0,0)+E(0,1)-E(1,0)+E(1,1)|	&\leq 2	\,,\nonumber\\
\SCHSH{4}=|-E(0,0)+E(0,1)+E(1,0)+E(1,1)|	&\leq 2	\,.
\end{align}
All permutations of the $o^k_{s_k}$ and $k$ map Eq.~\eqref{eq:CHSH} to one of the four inequalities in Eq.~\eqref{eq:CHSH4}, so the four inequalities in Eq.~\eqref{eq:CHSH4} are the complete set of Bell inequalities equivalent to the standard CHSH inequality. This set is referred to as the class of CHSH Bell inequalities. For two parties, the correlation functions are consistent with a locally causal theory if and only if they satisfy all inequalities in the class of CHSH Bell inequalities~\cite{A.Fine:PRL:1982}.

To see that quantum mechanics predicts violations of the CHSH inequalities, first note that for a quantum state $\rho$ and observables $\O^k_{s_k} ={\Omega}^k_{s_k}\cdot\vec{\sigma}$, where $\vec{\sigma}= \left(\sigma_x, \sigma_y, \sigma_z \right)$ is the vector of Pauli matrices and
\begin{equation}\label{eq:measurements}
	{\Omega}^k_{s_k} = \left(\sin\theta^k_{s_k}\cos\phi^k_{s_k}, \sin\theta^k_{s_k}\sin\phi^k_{s_k},
\cos\theta^k_{s_k}\right)	\,,
\end{equation}
the correlation functions are
\begin{equation}\label{eq:general_correlations}
E(s_1,s_2) = \tr{\rho(\O^1_{s_1}\otimes\O^2_{s_2})}	\,.
\end{equation}
For the Bell state $\rho=\proj{\Phi^+}$ defined in Eq.~\eqref{eq:bell_state}, Eq.~\eqref{eq:general_correlations} becomes
\begin{align}\label{eq:singlet_correlation}
E(s_1,s_2) =	\cos\theta^1_{s_1}\cos\theta^2_{s_2}	+ \sin\theta^1_{s_1}\sin\theta^2_{s_2}\cos\left( \phi^1_{s_1} + \phi^2_{s_2} \right)	\,.
\end{align}

If, for example, the measurements correspond to the vectors
\begin{align}\label{eq:measurements_2_party}
\Omega^1_{0} &= \left(1, 0, 0\right)\,, &
\Omega^1_{1} &= \left(0, 1, 0\right)\,,	\nonumber\\
\Omega^2_{0} &= \frac{1}{\sqrt{2}}\left(1, 1, 0\right)\,,&
\Omega^2_{1} &= \frac{1}{\sqrt{2}}\left(-1, 1, 0\right)	\,,
\end{align}
then the corresponding correlation functions are
\begin{align}
E(0,0) &= \frac{1}{\sqrt{2}}\,,&
E(0,1) &= -\frac{1}{\sqrt{2}}\,, 	\nonumber\\
E(1,0) &= -\frac{1}{\sqrt{2}}\,,&
E(1,1) &= -\frac{1}{\sqrt{2}}	\,. 
\label{Eq:SpecificCorFn}
\end{align}
Substituting these into Eq.~\eqref{eq:CHSH4}, one finds that $\SCHSH{4}=2\sqrt{2}>2$, i.e., the CHSH inequality is violated.

\subsubsection{Bell inequality violations with no aligned directions}

From Eqs.~\eqref{eq:CHSH4} and \eqref{eq:singlet_correlation}, it is clear that the choices of measurements that generate CHSH-inequality-violating correlations must satisfy some constraints, i.e., the directions that correspond to the measurements must be aligned in particular ways. However, if the observers do not share a reference frame, directions satisfying such constraints can only be chosen probabilistically. If the measurement directions $\Omega^{k}_{s_k}$ are sampled according to the normalized measure $\text{d}\Omega^{k}_{s_k}$, then the observers will choose measurements that generate correlations inconsistent with any locally causal theory with probability
\begin{equation}\label{Eq:Dfn:ProbViolation2}
 p = \int\,f_{\mbox{\tiny CHSH}}\left(\{\Omega^1_{0},\Omega^1_{1},\Omega^2_{0},\Omega^2_{1}\}\right)\text{d}\Omega^{1}_{0}\text{d}\Omega^{1}_{1}\text{d}\Omega^{2}_{0}\text{d}\Omega^{2}_{1}	\,,
\end{equation}
where $f_{\mbox{\tiny CHSH}}\left(\{\Omega^1_{0},\Omega^1_{1},\Omega^2_{0},\Omega^2_{1}\}\right)$ is a function that returns 1 if the orientations $\{\Omega^1_{0},\Omega^1_{1},\Omega^2_{0},\Omega^2_{1}\}$ generate correlations that violate any of the CHSH Bell inequalities and 0 otherwise. The probability with which the observers generate correlations inconsistent with any locally causal theory depends on the way they choose their measurements, which in turn depends on how much they can align their reference frames. For example, if the observers can completely align their measurement bases, they can simply choose the measurements \eqref{eq:measurements_2_party}, so $p = 1$, i.e., they always generate correlations inconsistent with any locally causal theory.

When the observers cannot align their measurements at all and randomly choose both of their measurements independently and isotropically, then the probability that they will choose measurements that generate correlations violating one of the CHSH inequalities is $\approx28.3\%$~\cite{ycliang2010}. However, if the observers choose their measurements to be locally orthogonal, i.e.,
\begin{align}
\Omega^1_{0} \cdot \Omega^1_{1}	= 0\,,\quad
\Omega^2_{0} \cdot \Omega^2_{1}	= 0	\,,
\end{align}
then the probability of generating correlations that violate a CHSH inequality increases to $\approx41.3\%$~\cite{ycliang2010}. 

\subsubsection{Bell inequality violations with one aligned direction}

Consider another possible scenario, in which the observers can align one direction of their measurements, e.g., the $z$ direction of the Bloch sphere. Then each observer can choose two orthogonal measurements in the $xy$-plane, i.e., choose two angles $\phi_1$ and $\phi_2$ randomly according to a uniform distribution. The four corresponding measurements are
\begin{align}
\Omega^0_{0} &= \left(\cos\phi_1, \sin\phi_1, 0\right)\,,	\nonumber\\
\Omega^0_{1} &= \left(-\sin\phi_1, \cos\phi_1, 0\right)\,,	\nonumber\\
\Omega^1_{0} &= \left(\cos\phi_2, \sin\phi_2, 0\right)\,,	\nonumber\\
\Omega^1_{1} &= \left(-\sin\phi_2, \cos\phi_2, 0\right)	\,.
\end{align}
Substituting these into Eq.~\eqref{eq:singlet_correlation} and then Eq.~\eqref{eq:CHSH4} gives
\begin{align}\label{eq:2parties_PROM}
\SCHSH{1}&=2|\cos\left(\phi_1 + \phi_2\right) - \sin\left(\phi_1 + \phi_2\right)|, \nonumber\\
\SCHSH{2}&=\SCHSH{3}=0,\nonumber\\
\SCHSH{4}&=2|\cos\left(\phi_1 +  \phi_2\right) + \sin\left(\phi_1 + \phi_2\right)|	\,.
\end{align}
Using standard trigonometric identities, we see that the CHSH inequalities in Eq.~\eqref{eq:CHSH4} are satisfied if and only if 
\begin{equation}
|\cos{x}|\leq \frac{1}{\sqrt{2}}\,,\quad 
|\sin{x}|\leq \frac{1}{\sqrt{2}}	\,,
\end{equation}
where $x = \phi_1 + \phi_2 + \frac{\pi}{4}$. But one of these inequalities is violated unless $x=\frac{\pi}{4}$, so any choice of measurements (except for a set of measure zero) will violate one of two CHSH inequalities, i.e., $p = 1$. Therefore in order to choose measurements that generate correlations inconsistent with any locally causal theory, it is sufficient to perfectly align a single direction, namely the $z$ axis, and to check only two Bell inequalities.

\subsection{The general scenario}\label{sec:general}

We now generalize the two-party case outlined in the previous section to $N$ parties and determine the extent to which the probability of generating correlations inconsistent with locally causal theories depends on the alignment of the measurements of the $N$ parties. To this end, we consider the scenario (abstracted from the physical implementation) wherein a verifier prepares a large number of copies of the $N$-partite GHZ state (the GHZ state is chosen as it is a resource for obtaining maximum violations of some commonly used Bell inequalities~\cite{WWproveWWZB,scarani2001}),
\begin{equation}
	\GHZ = \frac{1}{\sqrt{2}}\left(\ket{\vec{\textbf{0}}_N} + \ket{\vec{\textbf{1}}_N}\right)
\end{equation}
where $\ket{\vec{\textbf{0}}_N}$ and $\ket{\vec{\textbf{1}}_N}$ denote states in which each of the $N$ qubits are prepared in the states $\ket{0}$ and $\ket{1}$ respectively. The verifier distributes 1 qubit from each copy to $N$ observers. As in the two-party case, each observer chooses two measurement bases, which corresponds to the $k^{th}$ observer choosing two directions $\Omega^k_{s_k}$, parametrized as in Eq.~\eqref{eq:measurements}, in the Bloch sphere, where $s_k\in\Z_2 =\{0,1\}$. Each observer measures their qubits, randomly choosing $s_k$ for each qubit. The observers then send a list of the measurement labels $s_k$ and outcomes $\pm1$ for each copy back to the verifier, who will use the lists to determine if the measurement outcomes are inconsistent with a locally causal theory.

In contrast to the typical scenario where $\{\Omega^k_{s_k}\}$ are \textit{fixed a priori} to some optimal measurement bases that give the maximal violation of a specific labeling of a Bell inequality, we now consider a scenario where the measurement bases/directions $\{\Omega^k_{s_k}\}$ are chosen \textit{randomly} according to some distribution, but are fixed throughout the experiment. Formally, if we treat the measurement directions $\{\Omega^k_{s_k}\}$ as random variables, we can define the probability, $p^N_\I$, that the verifier identifies that the correlation functions are incompatible with the class of Bell inequalities $\I$ as
\begin{equation}\label{Eq:Dfn:ProbViolation}
 p^N_{\mathcal{I}} = \int\,f^N_\I\left(\{\Omega^k_{s_k}\}\right)
 \prod_{k=1}^{N}\prod_{s_k \in\Z_2}{\rm d}\Omega^{k}_{s_k}\,,
\end{equation}
where d$\Omega^{k}_{s_k}$ is the normalized measure associated with the sampling of measurement direction $\Omega^k_{s_k}$, and $f^N_\I\left(\{\Omega^k_{s_k}\}\right)$ is a function that returns 1 if the measurements $\{\Omega^k_{s_k}\}$ give rise to correlation functions that violate any of the Bell inequalities in the class $\mathcal{I}$ and 0 otherwise.

Clearly, $p^N_\I$ depends crucially on both the sampling of $\{\Omega^k_{s_k}\}$, which determines the probability of generating nonlocal correlations, and the class of Bell inequalities $\I$ involved, which determines the probability of the verifier \textit{identifying} nonlocal correlations as Bell-inequality-violating. For any given sampling of $\{\Omega^k_{s_k}\}$, $p_\I$ is thus upper bounded by $p_{\I_\text{all}}$=$\pNL$, where $\I_\text{all}$ refers to the complete set of Bell inequalities relevant to the scenario described above. The feasibility of demonstrating Bell inequality violation with randomly chosen measurement bases can then be quantified in terms of $\pNL$, which is the probability of randomly \textit{generating} correlations that are incompatible with any locally causal theory. We now discuss the method of identifying nonlocal correlations using an appropriate class of Bell inequalities.

\subsection{The Bell inequalities}
\label{sec:BI}

Bell inequalities are constraints on physically observable quantities that have to be satisfied by any locally causal theory~\cite{bell2004}. A relevant class of Bell inequalities for the scenario that we are considering is the set of $2^N$ Mermin-Ardehali-Belinski\v{\i}-Klyshko (MABK) inequalities~\cite{MABK1,MABK2}. A representative of this class is~\cite{ZBproveWWZB}
\begin{align}\label{eq:MABKold}
 \SN_1 &= \Bigl|\sum_{\vec{s} \in\Z_2^{\otimes
 N}}\beta\left(\vec{s}\right)E\left(\vec{s}\right)\Bigr|
 \le 2^{N}	\,,	
\end{align}
where $\vec{s} = \left(s_1,\ldots,s_N\right)$ is the vector of the $N$ measurement labels,
\begin{equation}\label{eq:betas}
 \beta\left(\vec{s}\right)=\!\!\!\sum_{\vec{a} \in\{-1,1\}^{\otimes N}}
 \sqrt{2}\cos\left[\frac{\pi}{4}(N+1-\sum_{j=1}^{N}
 a_j)\right]\prod_{l=1}^{N} a_l^{s_l}	\,,
\end{equation}
$\vec{a}= \left(a_1,\ldots,a_N\right)$ and the $N$-partite correlation functions $E\left(\vec{s}\right)$ are the expectation values of the product of the measurement outcomes when the $k^{th}$ observer performs the $s_k$-th measurement. Within quantum theory, the maximum possible value of $\SN_1$ is $2^{\frac{3N-1}{2}}$ \cite{MABK2, WWproveWWZB}.

As we show in Appendix A, this inequality can be rewritten as
\begin{equation}\label{Eq:MABK:S1}
 \mathcal{S}^N_{1}=\Bigl|\sum_{\vec{s} \in\Z_2^{\otimes N}}
 \beta\left(s,N\right)E\left(\vec{s}\right)\Bigr|
 \le2^N\,,
\end{equation}
where $s = \sum_{k=1}^N s_k$ and
\begin{equation}
 \beta\left(s,N\right) = 2^{\frac{N+1}{2}}\cos\left(\frac{\pi}{4}(1+N-2s)\right)	\,.	\label{eq:betanew}
\end{equation}

The equivalence class of MABK inequalities is the set of inequalities that can be obtained by permutating the measurement choices, $s_k$, measurement outcomes, $o_k$ and labeling of the observers $k$ in the coefficients $\beta(\vec{s})$ of inequality~\eqref{Eq:MABK:S1}. However, as we prove in Appendix~\ref{App:EquivalentMABK}, all such permuations can be obtained by permuting the measurement labels (i.e., $s_k\rightarrow 1-s_k$ for some set of $k\in\{1,2,\ldots,N\}$) and so each of the $2^N$ MABK inequalities can be obtained by one of the distinct $2^N$ permutations of measurement settings. In particular, the inequality
\begin{equation}\label{eq:MABK:S2}
 \mathcal{S}^N_2=\Bigl|\sum_{\vec{s} \in\Z_2^{\otimes N}}
 \beta\left(N-s,N\right)E\left(\vec{s}\right)\Bigr|
 \le2^N\,,
\end{equation}
which will play an important role in the scenario where a single direction is shared, can be obtained from inequality~\eqref{Eq:MABK:S1} via the mapping $s_k\rightarrow 1-s_k$ for all $k=1,2,\ldots, N$.

When $N=2$, the MABK inequalities reduce to the Bell-CHSH inequalities~\cite{CHSH} and represent the complete set of Bell inequalities for this scenario~\cite{A.Fine:PRL:1982}. So for $N=2$, the measurement outcomes are inconsistent with any locally causal theory if and only if they violate one of the MABK inequalities. For $N>2$, there are also other equivalence classes of Bell inequalities (see, for example, Refs.~\cite{C.Sliwa:PLA:165,WWproveWWZB,I.Pitowsky:PRA:014102}). An extensive set of such $N$-partite Bell inequalities that include the MABK class as a subset is the well-known Werner-Wolf-\.{Z}ukowski-Brukner (WWZB) inequalities~\cite{WWproveWWZB, ZBproveWWZB}. These $2^{2^N}$ Bell inequalities can be put into the form of the following single nonlinear Bell inequality
\begin{equation}\label{Ineq:nonlinear}
 S_\text{WWZB}=\sum_{\vec{a} \in\{-1,1\}^{\otimes N}}\Bigl|
 \sum_{\vec{s}\in\Z^{\otimes N}_2}
 \prod_{k=1}^N a_k^{s_k}E(\vec{s})\Bigr|\le2^{N}\,.
\end{equation}

Defining $\delta_N=1-N$~mod~2, the above inequality is both necessary and sufficient for the set of $2^N$ $N$-partite GHZ correlation functions [with measurements defined as in Eq.~\eqref{eq:measurements}]
\begin{equation}\label{Eq:GHZ:FullCorFn}
 E\left(\vec{s}\right)=\cos\left(\sum_{l=1}^N \phi^l_{s_l}\right)
 \prod_{k=1}^N\sin\,\theta^k_{s_k}+\delta_N\prod_{k=1}^N\cos\,\theta^k_{s_k}	\,,
\end{equation}
to be describable within a locally causal theory. However, not all measurement statistics are captured by these \textit{full} correlation functions. We can also compute the \textit{restricted} correlation functions of the GHZ state,
\begin{equation}\label{Eq:GHZ:ResCorFn}
 E\left(\{s_k\}_{k\in\K}\right)=
 \delta_{|\K|}\prod_{k\in\K}\cos\,\theta^k_{s_k}\,,
\end{equation}
$\K\subset\{1,2,\ldots,N\}$, which involve the expectation value of the product of the measurement outcomes for a subset of the $N$ parties. As a result, one generally needs to check the measurement statistics against the \textit{complete} set of Bell inequalities relevant to the particular experimental scenario to determine if these correlations are nonlocal. The characterization of the complete set of Bell inequalities is only known for $N=2$ and 3 (see Refs.~\cite{A.Fine:PRL:1982,C.Sliwa:PLA:165,I.Pitowsky:PRA:014102} for details).

Nevertheless, for small $N$, the problem of deciding if some given measurement statistics are compatible with a locally causal description can be solved numerically using \textit{linear programming}~\footnote{The set of locally causal correlations is a convex set with finitely many extreme points~\cite{I.Pitowsky:book}. To determine if some measurement statistics correspond to a member of this set, it suffices to check if the given measurement statistics can be written as a convex combination of these extreme points. For an alternative, but equivalent formulation of this problem as a linear program, see, for example, Refs.~\cite{D.Kaszlikowski:PRL:4418,M.B.Elliott:0905.2950} and the supplementary materials in Ref.~\cite{ycliang2010}.}. For larger values of $N$, it may become infeasible to compute $\pNL$ using linear programming. However, we can make use of the following inclusion relations:
\begin{equation}\label{Eq:BI:Inclusion}
 \{S_1^N\}\subset\{S_1^N, S_2^N\}\subset\text{MABK}\subseteq\text{WWZB}
 \subseteq\text{Complete Set}
\end{equation}
to lower bound this probability, i.e.,
\begin{equation}\label{Ineq:ProbViolation}
 p_{\{S_1^N\}}\le p_{\{S_1^N, S_2^N\}}\le \pM \le p_\text{WWZB}
 \le \pNL,
\end{equation}
where $p_\text{MABK}$ etc.\ are the probabilities defined in Eq.~\eqref{Eq:Dfn:ProbViolation} with $\mathcal{I}$ being the respective classes of Bell inequalities.

\subsection{Sampling of measurement directions}\label{sec:sampling}

Our sampling of measurement directions depends on the extent to which the $N$ observers are able to align their measurements within each physical scenario. For example, when all observers share a complete reference frame and can completely align their measurements, for any class of Bell inequalities $\I$, they can always pick $\{\Omega^k_{s_k}\}$ in such a way that $f^N_{\I}\left(\{\Omega^k_{s_k}\}\right)=1$, assuming there exist such measurements. In this paper, we assume that the observers either cannot align their measurements at all or can only align them with respect to a single direction $\vec{n}$. 

The following samplings of measurement directions will be applied to cater to the different extents in which the observers can align their measurements:
\begin{enumerate}
\item \textit{Random isotropic measurements (RIM)} -- each party $k$ chooses
 both directions $\Omega^{k}_{s_k}$ for
 $s_k=0,1$ independently and uniformly from the set of all possible
 directions;
\item \textit{Random orthogonal measurements (ROM)} -- each local pair of
 measurement directions is chosen to be orthogonal but otherwise
 uniform, i.e., RIM with the additional constraint:
 \begin{equation}\label{Eq:Dfn:ROM}
 	\Omega^k_0 \cdot \Omega^k_1 = 0 \quad \forall \ k\,;
 \end{equation}
\item \textit{Planar random orthogonal measurements (PROM)} -- in addition to Eq.~\eqref{Eq:Dfn:ROM}, all measurement directions are confined to a plane defined by some normal vector $\vec{n}$ (which corresponds to the common direction the observers can align), i.e.,
 \begin{equation}
 \Omega^{k}_{s_k}\cdot\vec{n}=0\quad \forall\ k,\ s_k\,,
 \end{equation}
 for some $\vec{n}$ shared by the $N$ parties.
\end{enumerate}

Some of the results presented in Sec.~\ref{Sec:Noiseless:ROM} for RIM and ROM have been discussed in Ref.~\cite{ycliang2010} but are included here in more detail.

\section{Noiseless scenarios}\label{sec:Noiseless}

When the $N$ experimenters do not align their measurements, one may expect that it is unlikely to find Bell-inequality-violating correlations by performing measurements in randomly chosen bases. Nonetheless, for the $N$-partite GHZ state, the probability of choosing measurements that violate a Bell inequality rapidly approaches 1 as $N$ increases. In Sec.~\ref{Sec:Noiseless:ROM} we briefly summarize the results for RIM and ROM presented in Ref.~\cite{ycliang2010} and analyze the difference between testing the correlations against the WWZB inequalities and testing the correlations against the full set of Bell inequalities (obtained for small $N$ using linear programming).

Without any alignment of measurements, if the experimenters do not test the experimental statistics against a class of Bell inequalities that grows exponentially with $N$, then the probability of identifying that the correlations generated by the randomly chosen measurements are inconsistent with any locally causal theory decreases with $N$. However, if the observers can align the $z$ direction of their measurements, then we prove that for all $N$ they can always choose measurements that violate one of two Bell inequalities, namely, $\SN_1$ or $\SN_2$ from Eqns.~\eqref{Eq:MABK:S1}--\eqref{eq:MABK:S2}, by an amount that grows exponentially with $N$. We also numerically calculate the probability of violating four different classes of Bell inequalities and show that as the aligned direction is rotated away from the $z$ axis, the observers have to test their experimental statistics against more Bell inequalities in order to certify that the correlations generated by the randomly chosen measurements are inconsistent with any locally causal theory.

\subsection{No aligned directions - RIM and ROM}\label{Sec:Noiseless:ROM}

\begin{table*}[t!]
\begin{ruledtabular}
\begin{tabular}{c||c|c|c|c|c|c|c|c|c|}
& \multicolumn{4}{c|}{RIM} & & \multicolumn{4}{c|}{ROM} \\ \hline
$N$ & $p_{\{S_1^N\}}$ & $\pM$ & $\pW$
& $\pNL$ & & $p_{\{S_1^N\}}$ & $\pM$ & $\pW$
& $\pNL$ \\ \hline
2 &7.080\% & 28.319\% & 28.319\% & 28.319\% & & 10.326\% & 41.298\% & 41.298\% & 41.298\%	\\
3 & 1.328\% & 10.002\% & 13.313\% & 74.690\% & & 2.324\% & 18.150\% & 26.604\% & 96.207\% \\
4 & 0.972\% & 13.410\% & 23.407\% & 94.238\% & & 1.714\% & 25.500\% & 59.034\% & 99.976\% \\
5 & 0.644\% & 15.210\% & 25.675\% & 99.593\% & & 1.108\% & 29.733\% & 52.798\% & 100.000\%\footnote{There are instances where a randomly generated correlation is local, but our simulation indicates that this happens less than 1 in every $10^6$ times. } \\
6 & 0.428\% & 16.879\% & 31.235\% & 99.965\% & & 0.734\% & 34.442\% & 71.190\% & 100.000\%\footnote{Of the $10^6$ randomly generated probability distributions, we did not find one that admits a locally causal description. } \\
8 & 0.183\% & 19.085\% & 37.509\% & - 	 & & 0.312\% & 41.935\% & 80.420\% & - \\
10 & 0.077\% & 20.443\% & 42.254\% & - 	 & & 0.133\% & 47.968\% & 86.926\% & - \\
15 & 0.009\% & 22.037\% & 50.515\% & - 	 & & 0.017\% & 59.006\% & 95.204\% & - \\
\end{tabular}
\end{ruledtabular}
\caption{\label{tbl:ProbViolation} Probability of finding a Bell inequality violation from the $N$-partite GHZ state for the scenario where each party is allowed to perform binary projective measurements in two randomly chosen measurement bases according to either RIM (left) or ROM (right). The number of parties $N$ is given in the leftmost column. Then, to the right, we have, respectively, the probability of violating $S_1^N$, the $2^N$ MABK inequalities, the $2^{2^N}$ WWZB inequalities, and the complete set of Bell inequalities relevant to this scenario. Note that the probability of violation for each class of Bell inequalities is lower bounded by the corresponding entry to its left, as expected from Eq.~\eqref{Ineq:ProbViolation}.}
\end{table*}

A natural strategy that the $N$ experimenters can adopt is to each randomly choose two independent measurement bases $\Omega^k_{s_k}$ according to a uniform distribution on the surface of a sphere. As can be seen from Eq.~\eqref{eq:measurements}, this corresponds to each observer randomly choosing 4 angles $\theta^k_{s_k}$ and $\phi^k_{s_k}$ for $s_k\in\Z_2$, where $\phi^k_{s_k}$ are chosen from a uniform distribution on the interval $[0,2\pi]$ and $\theta^k_{s_k}$ from the interval $[0,\pi]$ with $p(\theta) = \frac{1}{2}\sin\theta$. When the observers restrict their measurements to be orthogonal to each other (i.e., when they sample according to ROM), then Eq.~\eqref{Eq:Dfn:ROM} fixes one of the angles. 

The probability of violating 4 classes of Bell inequalities, namely, $\{\SN_1\}$, the $2^N$ MABK inequalities, the $2^{2^N}$ WWZB inequalities and the complete set of Bell inequalities for two binary-outcome measurements at each site, are presented in Table~\ref{tbl:ProbViolation}. For $N=2$ (and only for $N=2$), the MABK, WWZB and full set of Bell inequalities are all identical to the set of CHSH inequalities, so the probability of violating each of these three classes coincide for both RIM and ROM. 

For $N\ge3$, these three classes of inequalities obey the strict inclusion relations given in Eq.~\eqref{Eq:BI:Inclusion}, and we see that the probabilities of violating these three different classes follow the strict inequalities given in Eq.~\eqref{Ineq:ProbViolation}. For the MABK and WWZB inequalities, which contain a number of inequalities that is exponential in $N$, the probability of violation $p_\I$ generally increases with $N$, except when $N$ increases from 2 to 3 and a few other cases for ROM. This increasing trend is even more pronounced for the complete set of Bell inequalities where $\pNL$ is found to be strictly increasing (up to the limit of our analysis).

The WWZB inequalities are necessary and sufficient conditions for the full $N$-partite correlation functions to be consistent with a locally causal theory. Given that the restricted correlation functions can be computed from the respective reduced density matrices of $\GHZ$ and are always separable, it may seem surprising that the WWZB inequalities fail to detect a significant fraction of the nonclassical correlations generated from the GHZ states. However, while the reduced density matrices of $\GHZ$ are separable and so can be modelled in a locally causal theory, there is an additional requirement: the locally causal models for the different reduced density matrices must be consistent, in that they must not contradict one another and must also reproduce the full correlation functions given in Eq.~\eqref{Eq:GHZ:FullCorFn}. The results given in Table~\ref{tbl:ProbViolation} show that as $N$ increases, it becomes increasingly difficult to find a locally causal model that could simultaneously reproduce Eq.~\eqref{Eq:GHZ:FullCorFn} and Eq.~\eqref{Eq:GHZ:ResCorFn}.

The results from Table~\ref{tbl:ProbViolation} also imply that detecting nonclassical correlations with a probability that increases with $N$ requires a set of Bell inequalities containing a number of inequalities that is exponential in $N$. If we only use one MABK inequality, e.g., $\SN_1$, to detect these nonclassical correlations, then the probability of finding correlations that violate $\SN_1$ via ROM decreases exponentially as $N$ increases. Clearly, because each inequality in the same equivalence class can be obtained by adopting a different classical labeling, the probability of violating any one of the MABK inequalities is equal to $p_{\SN_1}$. Therefore the probability of violating one of a set of $M$ MABK inequalities is upper bounded by $M p_{\SN_1}$. As $p_{\SN_1}$ decreases exponentially with $N$, $M$ must increase exponentially with $N$ in order for the probability of violating one of a set of $M$ inequalities to either remain constant or increase. 

As we will demonstrate in the next section, this is not the case if the $N$ experimenters can align one of their measurement directions. In particular, we will show that if observers can align a measurement direction, then there is a set of two inequalities such that the probability of violating either of these two inequalities is one for all $N$.

\subsection{Partially aligned measurements - random measurements in the $xy$ plane}\label{sec:prom_xy}

Without any alignment of their measurements, observers need to check their experimental statistics against an exponentially large class of Bell inequalities to uncover nonclassical correlations with a probability that increases with the number of observers. However, there are physical situations in which it is relatively easy to align a single measurement direction, or such a direction is naturally defined by the system.

For example, if qubits are encoded in the polarization of single photons and transmitted over optical fibres, then the ordinary and extraordinary modes are stable but optical birefringence causes a phase shift between the two modes. If this phase shift is unknown, then the observers share a single `direction' on the Bloch sphere but have an essentially random alignment of the other two directions. While experimental techniques are available to account for this phase shift and may exist for other situations in which there is a preferred direction, we show that if the observers are trying to violate a Bell inequality, then such techniques are unnecessary (the related question for quantum key distribution in this situation has also been investigated~\cite{alaing2010}).

Specifically, we show that if the reference direction, $\vec{n}$, is the $z$-axis and the observers agree on a labeling convention for their measurements, they will always obtain correlation functions that violate either $\SN_1$ or $\SN_2$ if the measurements are sampled according to PROM (i.e., if the measurements are orthogonal and confined to the plane perpendicular to some normal vector $\vec{n}$ shared by the $N$ parties).

For PROM in the $xy$ plane, the observers share the $z$ axis. If the $k^{th}$ observer's two measurements are $\Omega^k_{0}$ and $\Omega^k_1$, then, because the labels $0$ and $1$ are arbitrary, they are free to relabel them as necessary so that $\{\Omega^k_0,\Omega^k_1,z\}$ forms a right-handed coordinate system for all $k$ (a similar result follows for left-handed coordinate systems). Under this labeling convention, randomly choosing $\Omega^k_{0}$ and $\Omega^k_1$ is equivalent to randomly choosing a single random angle $\chi_k$, with $\theta^k_{s_k}=\frac{\pi}{2}$ and
\begin{equation}\label{eq:labeling_convention}
\phi^k_{s_k} = \chi_k + s_k \frac{\pi}{2}	\,.
\end{equation}

\begin{thm}\label{thm:xyPROM}
Any choice of orthogonal measurements in the $xy$ plane on the $N$-partite GHZ state will generate correlation functions that satisfy either
\begin{align}
\SN_1 &= 2^{\frac{N+1}{2}}\Bigl|\sum_{\vec{s}\in\Z_2^{\otimes N}}
 \cos\left((1+N-2s)\frac{\pi}{4}\right)E(\vec{s})\Bigr| \geq 2^{\frac{3N}{2} - 1}	\,,\nonumber\\
\SN_2 &= 2^{\frac{N+1}{2}}\Bigl|\sum_{\vec{s}\in\Z_2^{\otimes N}}
 \cos\left((1-N+2s)\frac{\pi}{4}\right)E(\vec{s})\Bigr| \geq 2^{\frac{3N}{2} - 1}	\,,
 \label{Eq:S12N}
\end{align}
provided the observers obey the labeling convention described above.
\end{thm}

\begin{proof}
For the $N$-partite GHZ state and the labeling convention in Eq.~\eqref{eq:labeling_convention}, the full correlation
function, Eq.~\eqref{Eq:GHZ:FullCorFn}, becomes
\begin{align}
E\left(\vec{s}\right)= \cos\left(\sum_{k=1}^N \phi^k_{s_k}\right)
= \cos\left(\chi +s\frac{\pi}{2}\right)	\label{eq:correlation}	\,,
\end{align}
where $\chi = \sum_{k=1}^N \chi_k$ and as before, $s=\sum_{k=1}^N s_k$.
Substituting this into the left-hand-side of inequality~\eqref{Eq:MABK:S1} gives
\begin{align}
 \SN_1= 2^{\frac{N+1}{2}}\Bigl|\sum_{\vec{s}\in\Z_2^{\otimes N}}
 \cos\left((1+N-2s)\frac{\pi}{4}\right)\cos\left(\chi
 +s\frac{\pi}{2}\right)\Bigr|\label{eq:inequality1a}	\,.
\end{align}
There are $\binom{N}{s}$ ways of choosing $\vec{s}$ such that $\sum_{k=1}^N
s_k=s$, so Eq.~\eqref{eq:inequality1a} can be rewritten as
\begin{align}
 \SN_1 &= 2^{\frac{N+1}{2}}\Bigl|\sum_{s=0}^{N}\binom{N}{s}\cos\left((1+N-2s)\frac{\pi}{4}\right)
 \cos\left(\chi +s\frac{\pi}{2}\right)\Bigr|\nonumber\\
 &=2^{\frac{3N - 1}{2}}\Bigl|\sin\left(\chi + (N-1)\frac{\pi}{4}\right)\Bigr|	
 \label{eq:inequality1b}	\,.
\end{align}
Similarly, substituting Eq.~\eqref{eq:correlation} into the left-hand-side of
inequality~\eqref{eq:MABK:S2} gives
\begin{align}
 \SN_2 &= 2^{\frac{3N - 1}{2}}\Bigl|\sin\left(\chi + (N+1)\frac{\pi}{4}\right)\Bigr|	
 \label{eq:inequality2a}	\,.
\end{align}
{Because
\begin{equation}
 \max\left\{\left|\sin\left(x-\frac{\pi}{4}\right)\right|,
 \left|\sin\left(x+\frac{\pi}{4}\right)\right|\right\}\ge\frac{1}{\sqrt{2}}	\quad \forall \ x	\,,
\end{equation}
either} $\SN_1\geq2^{\frac{3N}{2} - 1}$ or $\SN_2\geq 2^{\frac{3N}{2} - 1}$.
\end{proof}

From Sec.~\ref{sec:BI}, the inequalities in Theorem~\ref{thm:xyPROM} are Bell inequalities with an \textit{upper} bound of $2^N$ in any locally causal theory. Therefore for $N=2$ and $\chi \neq \frac{k\pi}{2}$ for $k\in\Z$ the observers will violate $\SN_1$ or $\SN_2$. When $N>2$, the observers will always violate $\SN_1$ or $\SN_2$ by a factor of at least $2^{\frac{N}{2} - 1}$. Moreover, the upper bound for both $\SN_1$ and $\SN_2$ in quantum mechanics is $2^{\frac{3N-1}{2}}$, so the violation of $\SN_1$ or $\SN_2$ is within a constant factor $\frac{1}{\sqrt{2}}$ of the maximum violation possible in quantum mechanics.

We can also find the probability {$p\Bigl(\max\{\SN_1,\SN_2\}\geq (1-\epsilon)2^{\frac{3N-1}{2}}\Bigr)$} of the observers choosing measurements by PROM in the $xy$ plane such that they will obtain a violation of a Bell inequality that is within a factor $(1-\epsilon)$ of the maximum violation possible in quantum mechanics. 

Randomly choosing measurements by PROM in the $xy$ plane is equivalent to randomly choosing $\chi$ in Eqs.~\eqref{eq:inequality1b} and \eqref{eq:inequality2a}. The probability of choosing $\chi\in[0,2\pi]$ such that \begin{equation}
 \max\left\{\SN_1,\SN_2\right\}\ge (1-\epsilon)2^{\frac{3N-1}{2}}
\end{equation}
is the same as the probability of choosing $x\in[0,\frac{\pi}{4}]$ such that $\cos x\ge 1-\epsilon$, which is simply $\frac{4}{\pi}\cos^{-1}(1-\epsilon)$.

Therefore the probability of the observers choosing measurements by PROM in the $xy$ plane such that they will obtain a violation of either $\SN_1$ or $\SN_2$ above some threshold $(1-\epsilon)2^{\frac{3N-1}{2}}$ is
\begin{align}\label{eq:prob_threshold}
p\Bigl(\max\{\SN_1,\SN_2\}\geq (1-\epsilon)2^{\frac{3N-1}{2}}\Bigr) &= \frac{4}{\pi}\cos^{-1} (1-\epsilon) \,.
\end{align}

\begin{table*}[t!]
\begin{ruledtabular}
\begin{tabular}{l|| c |c|c|c|c|c|c|c|c|}
& \multicolumn{4}{c|}{$A_0$} & & \multicolumn{4}{c|}{$A_1$} \\ \hline
$N$ & $\{\SN_1,\SN_2\}$ & MABK & WWZB & Complete set  & & $\{\SN_1,\SN_2\}$ & MABK & WWZB & Complete set  \\ \hline
2 & 0.5411 & 0.7002 & 0.7002 & 0.7002 & & 0.0033 & 0.0033 & 0.0033 & 0.0033	\\
3 & 0.4129 & 0.4129 & 0.4580 & 0.9553 & & 0.2189 & 0.2189 & 0.2406 & 0.4850 \\
4 & 0.4729 & 0.4729 & 0.9130 & 0.9996 & & 0.3219 & 0.3219 & 0.3721 & 0.8189 \\
5 & 0.4832 & 0.4832 & 0.5867 & 0.9998 & & 0.3741 & 0.3741 & 0.4059 & 0.8635 \\
6~\footnote{Due to the small sample size for N = 6, we expect the corresponding entries of $A_0$ and $A_1$ to only be, respectively, lower and upper bounds.}  & 0.5035 & 0.5035 & 0.9155 & 0.9997 & & 0.4129 & 0.4129 & 0.4544 & 0.9782\\
\end{tabular}
\end{ruledtabular}
\caption{\label{tbl:rotated} Fraction of the surface area of a sphere that corresponds to normal vectors for which the probability of violating a class of Bell inequalities is nonzero ($A_0$) or unity ($A_1$). The four classes of Bell inequalities considered are $\{\SN_1,\SN_2\}$, the $2^N$ MABK inequalities, the $2^{2^N}$ WWZB inequalities and the complete set of Bell inequalities relevant to this scenario. Note that the value of both $A_0$ and $A_1$ exactly coincides for $\{\SN_1,\SN_2\}$ and the MABK inequalities when $N\neq 2$.}
\end{table*}

\begin{figure}[b!]
\centering
\includegraphics[width=\linewidth]{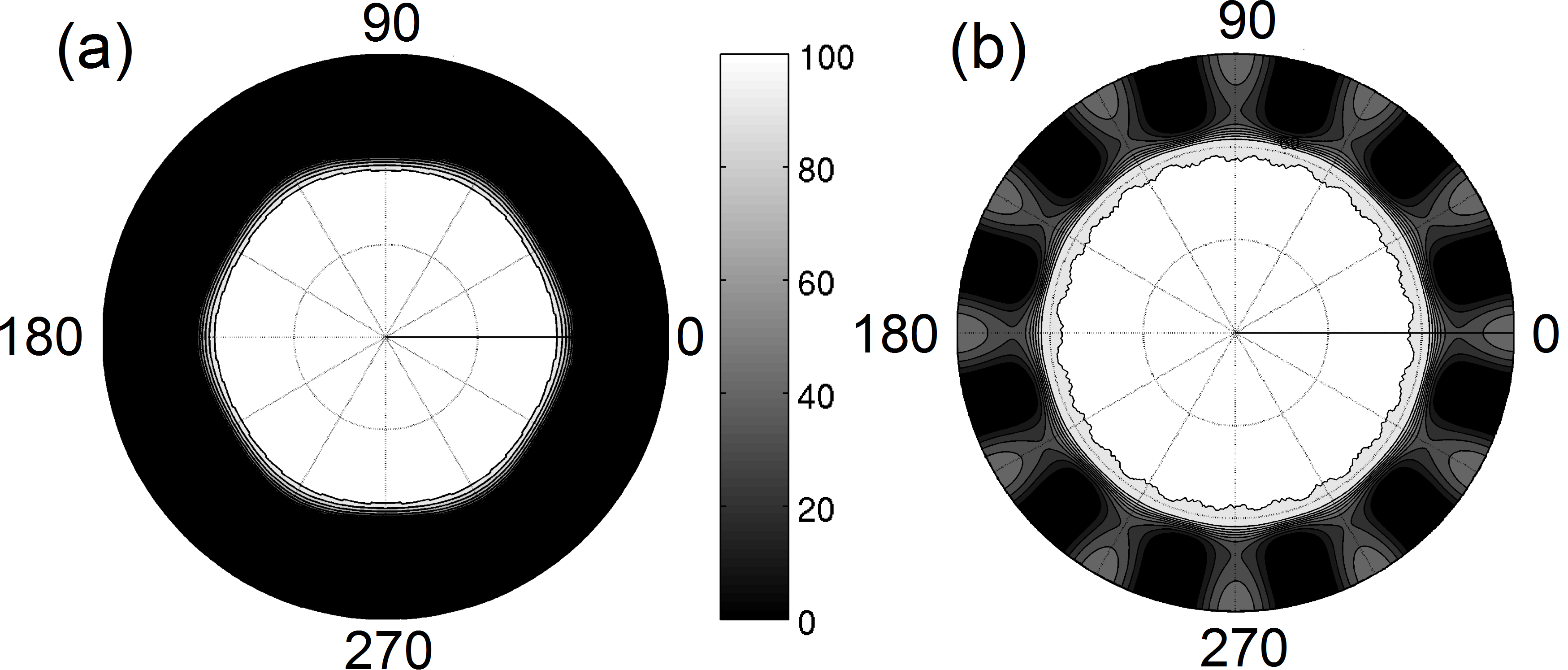}
\caption{(a) Contour plot of the probability of violating $\SN_1$ or $\SN_2$ for $N=6$ when sampling measurements via PROM with a reference direction as in Eq.~\eqref{Eq:Dfn:NormalVec}. (b) Contour plot of the probability of violating one of the WWZB inequalities for $N=6$ when sampling measurements via PROM with the reference direction as defined in Eq.~\eqref{Eq:Dfn:NormalVec}. In both contour plots, $\lambda = 0^{\circ}$ in the center and increases radially to a maximum of $90^{\circ}$ and $p_\I = 1$ for smaller $\lambda$ and generally decreases as $\lambda$ increases along a line of fixed $\alpha$ (i.e., along a radial line).}\label{Fig:PROM-rotated}
\end{figure}

\subsection{Partially aligned measurements - random measurements in other planes}\label{sec:prom_rotated}

Theorem~\ref{thm:xyPROM} applies when the observers measure along two orthogonal directions in the $xy$ plane, i.e., when the direction that the observers can align is the $z$-axis (which corresponds to the computational basis used to define the GHZ state). When the common direction is at some angle to the basis in which the GHZ basis is defined, the probability of observers obtaining correlation functions that violate a Bell inequality can change significantly. For $N\geq 2$, we simulate $p_\I(\vec{n})$ as a function of $\lambda$ and $\alpha$ where 
\begin{equation}\label{Eq:Dfn:NormalVec}
\vec{n} = \left(\cos\alpha \sin\lambda, \sin\alpha \sin\lambda, \cos\lambda\right) \,,
\end{equation}
and when $\I$ is $\SN_1$ and $\SN_2$, the $2^N$ MABK inequalities, the $2^{2^N}$ WWZB inequalities or the complete set of Bell inequalities for this scenario.

Given a normal vector, $\vec{n}$, shared by $N$ parties, we want the probability that $\vec{n}$ allows the $N$ parties to violate a given class of Bell inequalities with probability 1 or with some nonzero probability. Consequently, we define the ratio of the set of normal vectors that give violations of the class $\I$ of Bell inequalities with probability $p_\I(\vec{n}) = 1$ to the set of all normal vectors (i.e., the set of points on the surface of a unit sphere with $z\geq 0$) by
\begin{align}
A_1 &= \frac{1}{2\pi}\int_0^{\frac{\pi}{2}}d\lambda\int_0^{2\pi}d\alpha \sin\lambda g^1_{\I}(\alpha,\lambda)	\,,
\end{align}
where
\begin{align}
g^1_{\I}(\alpha,\lambda) = \begin{cases} 1 &\text{if } p_\I(\vec{n}) = 1\,,
\\
0 &\text{otherwise}\,.
\end{cases}	
\end{align}
Similarly, we define the ratio of the set of normal vectors that give violations of the class $\I$ of Bell inequalities with probability $p_\I(\vec{n}) > 0$ to the set of all normal vectors by
\begin{align}
A_0 &= \frac{1}{2\pi}\int_0^{\frac{\pi}{2}}d\lambda\int_0^{2\pi}d\alpha \sin\lambda g^0_{\I}(\alpha,\lambda)	\,,
\end{align}
where
\begin{align}
g^0_{\I}(\alpha,\lambda) = \begin{cases} 1 &\text{if } p_\I(\vec{n}) > 0\,,
\\
0 &\text{otherwise}\,.
\end{cases}	
\end{align}
$A_1$ gives the fraction of the set of unit normal directions $\vec{n}$ such that $z\geq 0$ and observers who share $\vec{n}$ can always obtain correlations that violate a Bell inequality when they sample measurements using PROM. Likewise, $A_0$ gives the fraction of the set of normal directions $\vec{n}$ such that $z\geq 0$ and observers who share $\vec{n}$ can always obtain correlations that violate a Bell inequality with nonzero probability when they sample measurements using PROM. The values of $A_0$ and $A_1$ for $N=2,\ldots,6$ are given in Table \ref{tbl:rotated} and the probability of violating $\{\SN_1,\SN_2\}$ and the WWZB inequalities for $N=6$ is plotted in Fig.~\ref{Fig:PROM-rotated}.

The case when $N=2$ is exceptional because almost any rotation of the reference direction [except when $\alpha=0$ in Eq.~\eqref{Eq:Dfn:NormalVec}] reduces the probability of violating a Bell inequality to below 1. This occurs because there are dense sets of measurements that produce arbitrarily small violations of $\SN_1$ and $\SN_2$ when $\vec{n} = (0,0,1)$ and these measurements do not produce a violation of $\SN_1$ or $\SN_2$ (or any other Bell inequality) when the reference direction is rotated an arbitrarily small amount from the $z$ axis. 

However, for $N>2$, there are no such sets of measurements and so, as our results show, the reference direction can be rotated from the $z$ axis by some ``small" angle $\lambda$ (i.e., $\lambda \lesssim 35^{\circ}$) in any direction without reducing $p_\I$.

Numerically, we found that for any rotated reference vector $\vec{n}$, considering the full class of MABK inequalities provides no advantage over just considering $\SN_1$ and $\SN_2$, except when $N=2$, in which case the MABK inequalities form the complete set of Bell inequalities. Considering the full set of WWZB inequalities does increase the value of $A_1$, but not very substantially. However, testing against the full set of WWZB inequalities can substantially increase the value of $A_0$, i.e., the area of normal vectors for which $p_\I(\vec{n})>0$, as can be seen in Table~\ref{tbl:rotated}. Testing against the full set of WWZB inequalities also reveals a strong dependence on the parity of $N$, which occurs due to the $\delta_N$ term in Eq.~\eqref{Eq:GHZ:FullCorFn}.	

For $N>2$, the dependence on the azimuthal angle, $\alpha$, is small when testing against $\SN_1$ and $\SN_2$. In particular, $p_{\{\SN_1,\SN_2\}}(\vec{n}) = 1$ for all $\alpha$ and $\lambda\lesssim 35^{\circ}$. That is, for any $N\geq 3$, the reference direction can be rotated from the $z$ axis by at least $35^{\circ}$ in any direction without affecting the probability of generating correlations that violate one of two Bell inequalities, namely, $\SN_1$ and $\SN_2$. As $N$ increases, this threshold value of the polar angle $\lambda$ appears to increase slowly.

\section{Noisy scenarios}\label{Sec:Noisy}

So far, we have made use of various idealisations. We now examine what happens when some of these assumptions are relaxed. In Sec.~\ref{Sec:Decoherence}, we determine how depolarizing and dephasing noise upon the GHZ state reduce the probability of violating a Bell inequality when observers do not align any measurement directions or only align a single measurement direction. In Sec.~\ref{Sec:Perturbed} we analyze how the probability of violating a Bell inequality is affected by random perturbations in each observer's alignment of the common direction, i.e., when the observers cannot align their measurements perfectly.

\subsection{Decoherence}\label{Sec:Decoherence}

In order to study the ability of observers to violate a Bell inequality in the presence of noise, we consider depolarizing and dephasing noise. For simplicity, we assume that each qubit is transmitted over equally noisy, uncorrelated channels, so the noise for all qubits is described by a single parameter $\nu$, where $\nu=0$ corresponds to zero noise and $\nu=1$ corresponds to maximal noise (i.e., complete depolarizing or dephasing). We begin by outlining the correlation tensor formalism, which is a convenient method of examining the effect of uncorrelated noise. We then give a brief introduction to depolarizing and dephasing noise before presenting our results from numerical simulations on the probability of violating a Bell inequality $p_\I$ in the presence of noise.

\subsubsection{Correlation tensor formalism}\label{sec:correlation_tensor}

An arbitrary $N$-qubit state $\rho$ can be expanded in any basis of Hermitian operators acting on the Hilbert space $ \mathcal{H}_{2^N}=\left(\mathbb{C}^2\right)^{\otimes N}$. In particular, the $N$-fold tensor products of local Pauli operators
\begin{equation}
\Sigma_{\vec{a}}	= \otimes_{k=1}^N \sigma_{a_k}
\end{equation}
is one such basis; here $\vec{a}\in\Z_4^{\otimes N}$ is a string of $N$ indices, $\sigma_{a_k}\in\{\unit_2,\sigma_x,\sigma_y,\sigma_z\}$, and $\unit_2$ is the identity operator acting on $\mathbb{C}^2$.

Together with the orthogonality relation,
\begin{equation}
\tr{\Sigma_{\vec{a}}\Sigma_{\vec{b}}} =2^N\delta_{\vec{a},\vec{b}}	\,,
\end{equation}
we can then represent $\rho$ by
\begin{equation}\label{Eq:rho:Expansion}
\rho = \frac{1}{2^N}\sum_{\vec{a}\in\Z^{\otimes N}_4} T_{\vec{a}}\Sigma_{\vec{a}}	\,,
\end{equation}
where $T_{\vec{a}} = \text{Tr}\left[\rho\Sigma_{\vec{a}}\right]$	
is the \textit{correlation tensor}~\cite{laskowski2010}. The description in terms of the correlation tensor is thus equivalent to the description in terms of the density operator. In what follows, we will follow Ref.~\cite{laskowski2010} and describe the effect of noise on a quantum state using the correlation tensor, which allows us to define the effects of uncorrelated noise on each qubit. For the GHZ state, we have
\begin{align}
T_{\vec{a}} &= \text{Tr}\left[\GHZrho\Sigma_{\vec{a}}\right] \nonumber\\
&= \frac{1}{2}\bra{\vec{\textbf{0}}_N}\Sigma_{\vec{a}}\ket{\vec{\textbf{0}}_N} + \frac{1}{2}\bra{\vec{\textbf{1}}_N}\Sigma_{\vec{a}}\ket{\vec{\textbf{1}}_N}	\nonumber\\
&\quad + \frac{1}{2}\bra{\vec{\textbf{0}}_N}\Sigma_{\vec{a}}\ket{\vec{\textbf{1}}_N} + \frac{1}{2}\bra{\vec{\textbf{1}}_N}\Sigma_{\vec{a}}\ket{\vec{\textbf{0}}_N}	\,.
\label{Eq:T:GHZ}
\end{align}
All of these terms are $0$ unless $\Sigma_{\vec{a}}$ is a tensor product of either (1) $2k$ Pauli $y$ and $(N-2k)$ Pauli $x$ matrices or (2) $2k$ Pauli $z$ and $N-2k$ identity matrices for some $k\in\Z$.

\subsubsection{Depolarizing noise}\label{sec:depolarizing}

Depolarizing noise maps local Pauli operators as~\cite{nielsen2000}
\begin{align}
\unit_2 &\rightarrow \unit_2\,, &  \sigma_x &\rightarrow (1-\nu)\sigma_x\,,	\nonumber\\
\sigma_y &\rightarrow (1-\nu)\sigma_y\,, &  \sigma_z &\rightarrow (1-\nu)\sigma_z	\,.
\end{align}
Full correlation functions correspond to all observers performing non-trivial projective measurements, while restricted correlation functions correspond to some subset of observers performing the measurement $\unit_2$ (i.e., ignoring the outcomes from some observers). Therefore the effects of depolarization on the correlation functions are 
\begin{gather}\label{Eq:EffectOnCorFn:LocalDepolarizing}
	E(\vec{s})\to (1-\nu)^N E(\vec{s})	\,,
\end{gather}
and
\begin{gather}\label{Eq:EffectOnRestCorFn:LocalDepolarizing}
	E(\{s_k\}_{k\in\kappa})\to (1-\nu)^{|\kappa|} E(\{s_k\}_{k\in\kappa})	\,,
\end{gather}
for arbitrary subsets of observers, $\kappa\subset\{1,\ldots,N\}$.

Note that the effect on the full correlation functions is identical to the effect of mixing the GHZ state with the maximally mixed state $\unit_{2^N}$ according to
\begin{equation}
\GHZ \rightarrow (1-\mu)\GHZrho + \frac{\mu}{2^N}\unit_{2^N}
\end{equation}
when $(1-\mu)\rightarrow(1-\nu)^N$. 

\subsubsection{Dephasing noise}\label{sec:dephasing}

We also consider dephasing noise, which is appropriate when there is some preferred basis in the system which is particularly stable. Dephasing noise in the computational basis suppresses off-diagonal terms for each qubit, i.e., it maps local Pauli operators as~\cite{nielsen2000}
\begin{align}
\unit_2 &\rightarrow \unit_2\,, &  \sigma_x &\rightarrow (1-\nu)\sigma_x\,,	\nonumber\\
\sigma_y &\rightarrow (1-\nu)\sigma_y\,, &  \sigma_z &\rightarrow \sigma_z	\,.
\end{align}
Clearly, from Eqs.~\eqref{Eq:rho:Expansion} and \eqref{Eq:T:GHZ} all diagonal terms of $\proj{\Psi_N}$ are unaffected and, because off-diagonal terms of the correlation tensor are zero unless the term corresponds to tensor products of only $\sigma_x$ and $\sigma_y$ matrices, all off-diagonal terms are uniformly reduced by a factor of $(1-\nu)^N$. Therefore dephasing takes the GHZ state to
\begin{equation}
\frac{1}{2}(\ket{\vec{\textbf{0}}_N}\!\bra{\vec{\textbf{0}}_N} + \ket{\vec{\textbf{1}}_N}\!\bra{\vec{\textbf{1}}_N})+ \frac{(1-\nu)^N}{2}(\ket{\vec{\textbf{0}}_N}\!\bra{\vec{\textbf{1}}_N} + \ket{\vec{\textbf{1}}_N}\!\bra{\vec{\textbf{0}}_N})	\,.\label{Eq:LocalDephasing}
\end{equation}

For a dephased GHZ state, the full correlation functions $E\left(\vec{s}\right)$ of Eq.\eqref{Eq:GHZ:FullCorFn} are
\begin{equation}\label{Eq:GHZ:FullCorFn:Dephased}
 (1-\nu)^N\cos\left(\sum_{l=1}^N \phi^l_{s_l}\right) \prod_{k=1}^N\sin\,\theta^k_{s_k}+\delta_N\prod_{k=1}^N\cos\,\theta^k_{s_k}\,.
\end{equation}
As the restricted correlation functions, Eq.~\eqref{Eq:GHZ:ResCorFn}, depend only on the component of the measurements in the $z$ direction, they are unchanged under dephasing noise (note that this is a property specific to the GHZ state).

Much like its separability property, a decohered or dephased $\GHZ$ also gradually loses its ability to violate \textit{any} Bell inequality as the level of noise (characterised by $\nu$) increases. Some bounds on the levels of dephasing and depolarising noise at which the $\GHZ$ state no longer violates a Bell inequality can be found in Ref.~\cite{laskowski2010}.

\subsubsection{Measurements in all directions - RIM and ROM}
\label{Sec:Noisy:ROM}

\begin{figure*}[t!]
\centering
\includegraphics[width=\linewidth]{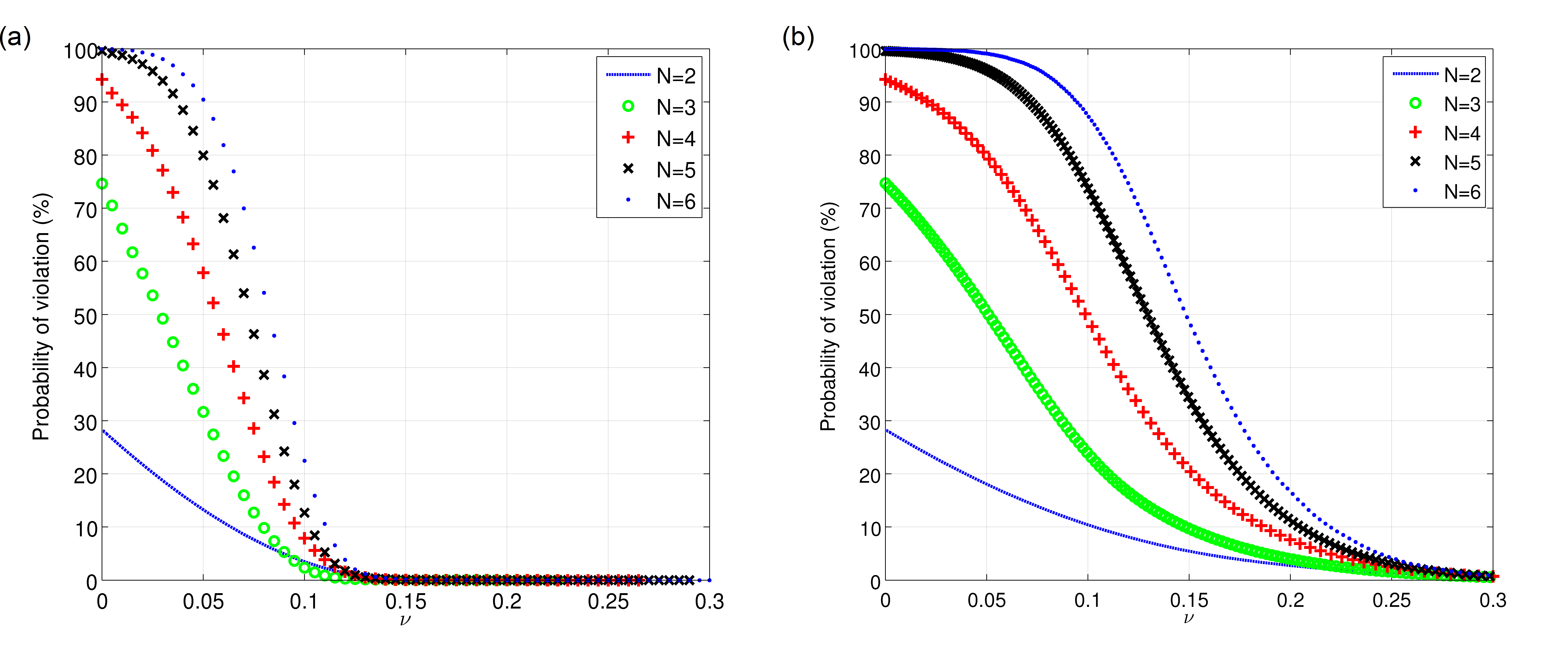}
\caption{\label{Fig:RIM-noisy} (Color online) Plot of the probability of violation $\pNL$, sampled using RIM for $N=2,\ldots,6$ observers, as a function of the noise parameter $\nu$ for (a): depolarizing noise and (b): dephasing noise.}
\end{figure*}

\begin{figure*}[t!]
\centering
\includegraphics[width=\linewidth]{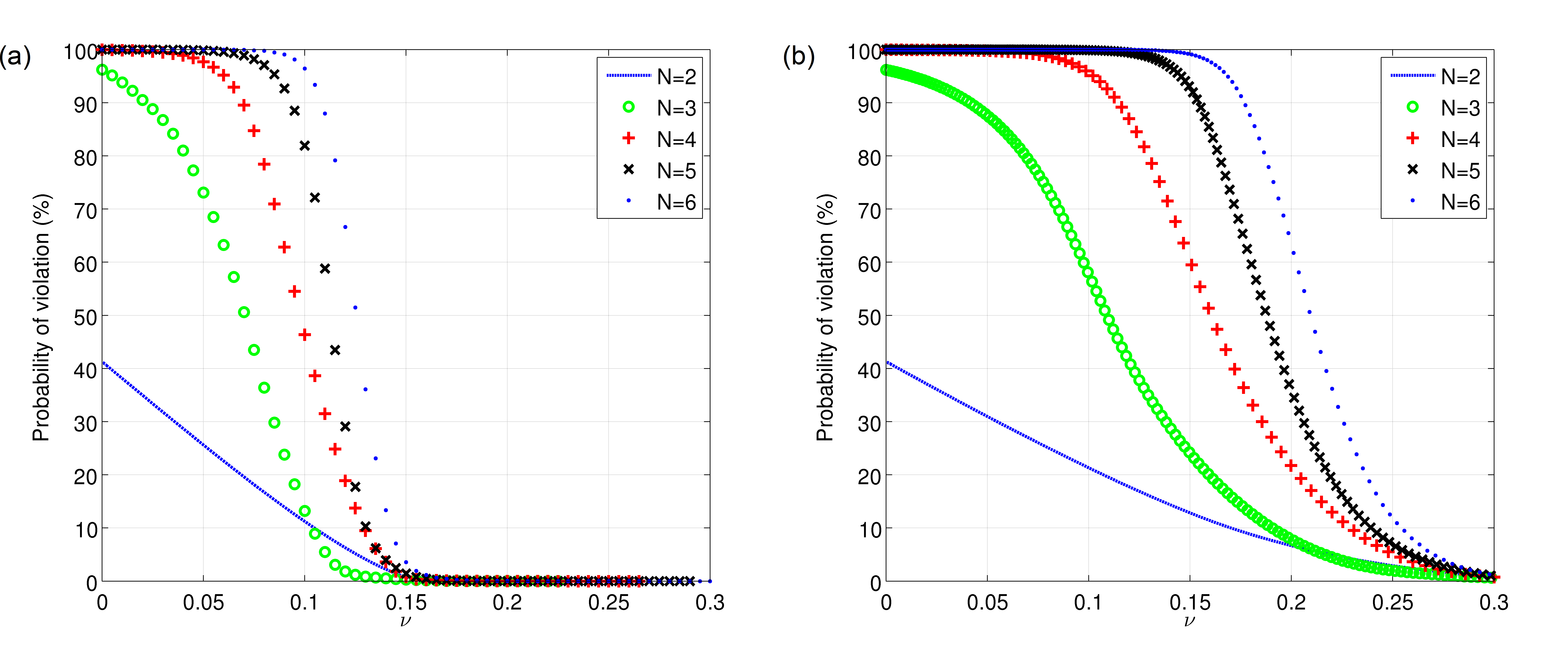}
\caption{\label{Fig:ROM-noisy} (Color online) Probability of violation $\pNL$, sampled using ROM for $N=2,\ldots,6$ observers, as a function of the noise parameter $\nu$ for (a): depolarizing noise and (b): dephasing noise.}
\end{figure*}

For $N=2,\ldots,6$ we have numerically calculated the probability of violating a Bell inequality, $\pNL$, under the influence of depolarizing and dephasing noise for RIM (Fig.~\ref{Fig:RIM-noisy}) and ROM (Fig.~\ref{Fig:ROM-noisy}). For both RIM and ROM, the probability of violating a Bell inequality for a given level of $\nu$ is always greater for dephasing noise than it is for depolarizing.

From these plots, we see that for all $N$, the probability of demonstrating Bell-inequality-violating correlations via either RIM or ROM on the GHZ state is robust against depolarizing and dephasing noise. Furthermore, this robustness increases with $N$. For example, with $10\%$ dephasing noise (i.e., $\nu = 0.1$), the probability of violation is reduced to $37\%$ (RIM) or $52\%$ (ROM) of the probability in the absence of noise in the case of the bipartite GHZ state whereas for the six-partite GHZ state, the probability of violation is reduced to $96\%$ (RIM) of the probability in the absence of noise or not affected to within the accuracy of the simulations (ROM). Furthermore, sampling measurements according to ROM not only increases the probability of violating a Bell inequality in the absence of noise compared to sampling measurements according to RIM, but also increases the stability with respect to both dephasing and depolarizing noise.

It is curious that for some range of noise parameters, the bipartite maximally entangled state actually gives a higher probability of violation as compared with the tripartite GHZ state.

\subsubsection{Partially aligned measurements - PROM in the $xy$ plane} \label{Sec:Noisy:PartialRFF}

When the observers choose orthogonal measurements in the $xy$ plane (i.e., $\theta= \frac{\pi}{2}$), they will obtain correlation functions that violate one of two MABK inequalities by an exponential amount. As we now show, this exponential violation of a MABK inequality translates directly into stability with respect to depolarizng and dephasing noise.

Let us now consider the effect of noise. For measurements on the $xy$ plane, Eq.~\eqref{Eq:GHZ:ResCorFn} implies that all restricted correlation functions vanish. It then follows from Eqs.~\eqref{Eq:EffectOnCorFn:LocalDepolarizing} and \eqref{Eq:GHZ:FullCorFn:Dephased} that the effects of depolarizing and dephasing noise are equivalent and both smoothly reduce $S^N_k$ to $(1-\nu)^N S^N_k$. Therefore the observers will always violate one of two MABK inequalities with PROM in the $xy$ plane if
\begin{equation}
(1-\nu)^N 2^{\frac{3N}{2}-1} > 2^N	\quad \Rightarrow \quad \nu < 1 - \frac{\sqrt[N]{2}}{\sqrt{2}}	\,.	
\label{eq:upperbound}
\end{equation}
Moreover, since the maximum MABK violation attainable by the GHZ state is $\sqrt{2}^{N-1}$ times the classical upper bound~\cite{WWproveWWZB}, correlations generated from the noisy GHZ state, Eq.~\eqref{Eq:EffectOnCorFn:LocalDepolarizing}, will never violate any MABK inequality when
\begin{equation}
\nu \ge  1-\frac{\sqrt[2N]{2}}{\sqrt{2}}	\,.
\label{eq:lowerbound}
\end{equation}
Therefore, as with the scenario where measurements are not restricted to a plane, the ability of observers to always violate either $\SN_1$ or $\SN_2$ is increasingly robust against decoherence as the number of observers increases. For $N\rightarrow \infty$, these limits are both $1-\frac{1}{\sqrt{2}}$ and observers will (depending on the level of noise $\nu$) either violate $\SN_1$ or $\SN_2$ for any choice of measurements or never violate $\SN_1$ or $\SN_2$.

The above analysis gives the critical value of $\nu$ at which the probability of violating $\SN_1$ or $\SN_2$ reaches the extremal values $1$ and $0$. For intermediate noise levels, namely, for $\nu\in[1-\frac{\sqrt[N]{2}}{\sqrt{2}},1-\frac{\sqrt[2N]{2}}{\sqrt{2}}]$, we can use Eq.~\eqref{eq:prob_threshold} to calculate the probability of violation. From Eqs.~\eqref{eq:prob_threshold} and \eqref{Eq:EffectOnCorFn:LocalDepolarizing}, we have
\begin{align}
 &p\left(\max\{\SN_1(\nu),\SN_2(\nu)\}>(1-\nu)^N(1-\epsilon)2^{\frac{3N-1}{2}}\right) \nonumber\\
 = &\frac{4}{\pi}\cos^{-1}(1-\epsilon)	\,.\label{eq:noiseprobability}
\end{align}
Therefore the observers will obtain statistics that violate $\SN_1$ or $\SN_2$ with probability $\frac{4}{\pi}\cos^{-1}(1-\epsilon)$ if
\begin{align}
 (1-\nu)^N(1-\epsilon)2^{\frac{3N-1}{2}} &= 2^N \nonumber \\
 \Rightarrow 1-\epsilon &= \frac{2^{\frac{1-N}{2}}}{(1-\nu)^N}	\,.	\label{eq:noiseviolation}
\end{align}
Substituting Eq.~\eqref{eq:noiseviolation} into Eq.~\eqref{eq:noiseprobability}, we find that with PROM on the $\nu$-locally-dephased-GHZ state, the observers will violate $\SN_1$ or $\SN_2$ with probability
\begin{equation}
p_{\{S_1^N, S_2^N\}}(\nu) =\frac{4}{\pi}\cos^{-1}\left(\frac{2^{\frac{1-N}{2}}}{(1-\nu)^N}\right)
\end{equation}
for $\nu\in[1-\frac{\sqrt[N]{2}}{\sqrt{2}},1-\frac{\sqrt[2N]{2}}{\sqrt{2}}]$.

\subsection{Imperfectly aligned measurement directions}\label{Sec:Perturbed}

\begin{figure}[t!]
\centering
\includegraphics[width=0.8\linewidth]{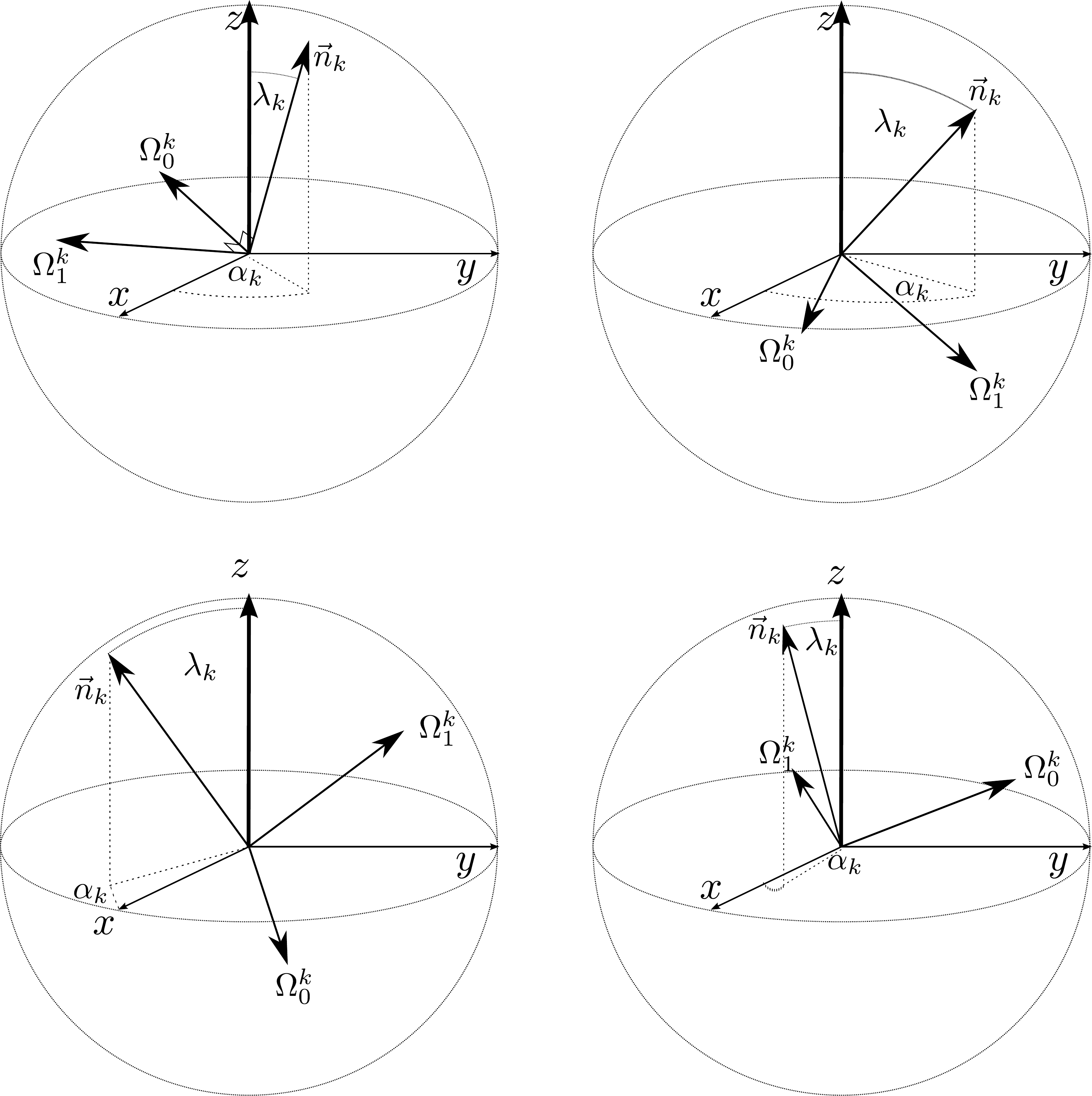}
\caption{\label{fig:bloch_normal} Representation of the normal ($\vec{n}_k$) and measurement directions ($\Omega^k_{s_0}$ and $\Omega^k_{s_0}$) for several parties.}
\end{figure}

So far, when discussing PROM, we have assumed that the observers have a perfectly aligned direction. We now consider the case where each observer may have some local approximation
\begin{equation}
\vec{n}_k = \left(\cos\alpha_k \sin\lambda_k, \sin\alpha_k \sin\lambda_k, \cos\lambda_k\right) \label{Eq:Dfn:NormalVecInd}
\end{equation}
to the $z$ axis (i.e., the basis in which the GHZ state is defined), see Fig.~\ref{fig:bloch_normal}. We assume that the azimuthal angles $\alpha_k\in[0,2\pi]$ are distributed uniformly and the $\lambda_k\in[0,\frac{\pi}{2}]$ are distributed such that
\begin{equation}
 p(\lambda^k_{s_k}) = (1+\frac{2}{\lambda_{\rm std}^2}) [\cos{\frac{\lambda_k}{2}}]^{\frac{4}{\lambda_{\rm std}^2}}	\,.
\end{equation}
Note that $\vec{n}_k$ with $\lambda_k >\frac{\pi}{2}$ are equivalent to $\vec{n}_k$ with $\lambda_k \leq\frac{\pi}{2}$, but will change the handedness of the labeling convention discussed in Sec.~\ref{sec:prom_xy} for the $k^{th}$ observer, thus changing which Bell inequality will be violated. For sufficiently small $\lambda$, the distribution is analogous to a Gaussian distribution with mean 0 and standard deviation $\lambda_{\rm std}$ on the surface of a sphere~\cite{bartlett2007}.

\begin{figure*}[t!]
\centering
\includegraphics[width=\linewidth]{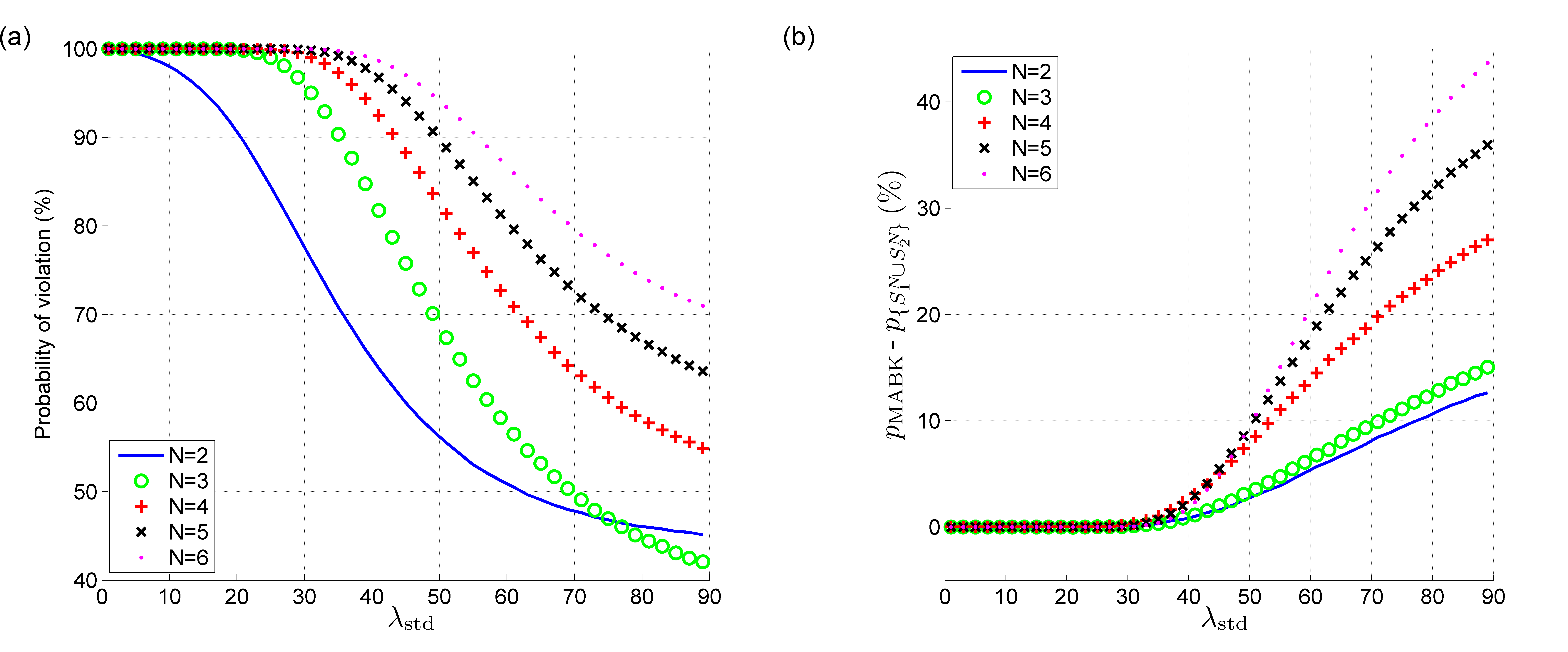}
\caption{\label{Fig:ROM-perturbed-MABK} (Color online) (a) Probability of violation $\pM$ when measurements are sampled using orthogonal measurements in the plane perpendicular to the normals $\vec{n}_k$ that are distributed according to a pseudo-Gaussian distribution as a function of the standard deviation of the distribution of polar angles $\lambda_\text{std}$ for $N=2,\ldots,6$ observers. (b) Increase in the probability of observers finding that their correlation functions do not correspond to a locally causal model when they check against the the full set of $2^N$ MABK inequalities, rather than just $\SN_1$ and $\SN_2$}.
\end{figure*}

The $k^{th}$ observer then measures in the plane perpendicular to $\vec{n}_k$, i.e., the $k^{th}$ observer's two measurements are now
\begin{align}
(\Omega^k_{s_k})_x &= \sin\phi^k_{s_k}\cos\lambda_k\cos\alpha_k + \cos\phi^k_{s_k}\sin\alpha_k \nonumber \\
(\Omega^k_{s_k})_y &= \sin\phi^k_{s_k}\cos\lambda_k\sin\alpha_k - \cos\phi^k_{s_k}\cos\alpha_k \nonumber \\
(\Omega^k_{s_k})_z &= -\sin\phi^k_{s_k}\sin\lambda_k \,.
\end{align}
Each observer still chooses random orthogonal measurements and the same labeling as before, i.e.,
\begin{equation}
\phi^k_{s_k} = \chi_k + s_k \frac{\pi}{2}	\,.
\end{equation}

In Fig.~\ref{Fig:ROM-perturbed-MABK} we present numerical results for the probability of violating one of the $2^N$ MABK inequalities as a function of the standard deviation of the distribution of polar angles $\lambda_\text{std}$ for $N=2,\ldots,6$ observers and the difference that checking the full set of $2^N$ MABK inequalities rather than just $\SN_1$ and $\SN_2$. Even for standard deviations $\lambda_\text{std}\approx 90^{\circ}$, the probability of violating an MABK inequality is higher than when the observers do not share a direction, {\it cf.} Table~\ref{tbl:ProbViolation}. The effect of increasing $\lambda_\text{std}$, that is, of decreasing the average accuracy of the local approximation to the common direction, is to smoothly decrease the probability of violating a MABK inequality. We find that checking against more inequalities makes little difference for $\lambda_\text{std}\lesssim 30^{\circ}$ (i.e., the probability increases by less than $0.1\%$). As the individual approximations become less accurate, the difference increases approximately linearly and also increases with $N$. At first glance, this seems rather counter-intuitive, as for $\lambda_\text{std}\approx 90^{\circ}$, the observers have essentially no idea what the reference direction is. The difference arises because the pseudo-Gaussian distribution is biased to smaller values of $\lambda$, which is only defined modulo $90^{\circ}$. 

\section{Discussion and Conclusion}
\label{sec:conclusion}

In this paper, we have shown that the degree to which $N$ observers can align their measurements substantially affects the probability of them generating Bell-inequality-violating (BIV) correlations by performing randomly chosen measurements on the $N$-partite GHZ state. Furthermore, the difficulty involved in verifying that some correlations are BIV is also closely related to the extent to which the observers can align their measurements: the better alignment they have, the fewer inequalities are needed to verify the nonlocal nature of these correlations. However, this reduction in difficulty does depend on how the observers can align their measurements. If, for example, they align the measurement directions corresponding to the basis in which the shared GHZ state is defined, then the observers will always generate BIV-correlations and this can be verified by testing the correlation functions against just two fixed Bell inequalities. As the aligned direction is rotated away from the $z$ axis (which corresponds to the basis in which the GHZ state is defined), the probability of violating a Bell inequality smoothly decreases. These results may make it easier to test Bell inequality violations over large distances, as they reduce the need to align distant measurements.

We have also shown that these results and the results presented in Ref.~\cite{ycliang2010} are strongly robust against uncorrelated noise. Moreover, we have shown that even if the observers can only partially align a measurement direction, i.e., if each observer has an approximation to the $z$ axis that is distributed with a standard deviation of up to $35^{\circ}$, they can still obtain BIV-correlations with probability 1. This suggests that the idea of demonstrating BIV correlations using randomly chosen measurement bases is not only an idealization, but is also applicable to real-world scenarios. Our results may also have implications for reference-frame-independent quantum key distribution~\cite{alaing2010} in the presence of noise. Furthermore, given the close connection between Bell inequality violation and the security of quantum key distribution protocols, it will be interesting to investigate if the results presented here have any implications on real-life implementation of device-independent quantum key distributions~\cite{diqkd}. 

Our results also provide insight into the structure of the set of locally causal correlations and its relations with respect to the set of quantum correlations. To this end, consider the results presented in Table~\ref{tbl:ProbViolation}. As discussed in Sec.~\ref{Sec:Noiseless:ROM}, the probability of violating any one of the MABK inequalities is equal. If a given set of correlation functions could violate at most one MABK inequality, then we would have
\begin{equation}\label{Eq:AtMost1Ineq}
R := \log_2\frac{\pM}{p_{S_1^N}} = N	\,,
\end{equation}
as $\SN_1$ is one of the $2^N$ MABK inequalities. For $N=2$, this holds, and one can indeed show that at most one CHSH inequality can be violated (see the supplementary material of Ref.~\cite{ycliang2010}). However, for larger values of $N$, we see that Eq.~\eqref{Eq:AtMost1Ineq} no longer holds: for example, with RIM, we have $R = 2.913$ and $R = 11.258$, respectively, for $N=3$ and $N=15$; similar scaling is also observed for ROM. In relation to this, it will also be interesting to determine if there are other aspects about the sets of correlations that we can learn by surveying randomly generated correlations, a problem that we shall leave for future research.

\begin{acknowledgments}
YCL acknowledges helpful discussion with Nicolas Gisin, Jean-Daniel Bancal, Tomasz Paterek, and Stephanie Wehner. SDB acknowledges helpful discussions with Terry Rudolph and Nick Harrigan. This work is supported by the Australian Research Council, the Swiss NCCR ``Quantum Photonics'' and the European ERC-AG QORE.
\end{acknowledgments}

\appendix
\section{Reformulation of the MABK inequalities}
\label{app:newbeta}

In this appendix, we show that the coefficients
\begin{equation}\label{eq:equivalentbetas}
 \beta\left(\vec{s}\right)=\!\!\!\sum_{\vec{a} \in\{-1,1\}^{\otimes N}}
 \sqrt{2}\cos\left[\frac{\pi}{4}(N+1-\sum_{j=1}^{N}
 a_j)\right]\prod_{l=1}^{N} a_l^{s_l}	\,,
\end{equation}
in the MABK inequalities depend only on $s=\sum_{j=1}^N s_k$ and $N$. In particular, we prove that
\begin{equation}\label{Eq:beta:simplified}
\beta\left(\vec{s}\right) =: \beta\left(s,N\right) = 2^{\frac{N+1}{2}}\cos\left(\frac{\pi}{4}(1+N-2s)\right)	\,,
\end{equation}
for all $\vec{s} \in \Z^{\otimes N}_2$ and $N\geq2$. The proof is by induction in $N$.

\begin{proof}
For the purpose of the proof, we will use $\vec{s}^{N-1}$ and $\vec{s}^N$ to denote the $(N-1)$-bit string $(s_1,\ldots,s_{N-1})$ and the $N$-bit string $(s_1,\ldots,s_N)$; likewise for $\vec{a}^{N-1}$ and $\vec{a}^N$. Moreover, let us define $s':=\sum_{k=1}^{N-1} s_k$, $a:=\sum_{k=1}^{N} a_k$, $a':=\sum_{k=1}^{N-1} a_k$ and
\begin{equation}
\gamma\left(\vec{s}^N\right)=\sum_{\vec{a} \in \{-1,1\}^{\otimes N} }\sqrt{2}\sin\left(\frac{\pi}{4}(N+1-a)\right)\prod_{l=1}^{N} a_l^{s_l}	\,. \label{eq:appoldgamma}
\end{equation}

For $N=2, 3$, it can be easily checked by explicit calculation that Eq.~\eqref{Eq:beta:simplified} holds. To prove Eq.~\eqref{Eq:beta:simplified} for general $N>2$, let us first establish some recursion relations. To this end, let us expand Eq.~\eqref{Eq:beta:simplified} in terms of $a_N=\pm1$ to get
\begin{align}
 \beta\left(\vec{s}^N\right) =
 &\sum_{\vec{a}^{N-1} \in \{-1,1\}^{\otimes N-1} }\sqrt{2}\cos\left(\frac{\pi}{4}(N - a')\right)\prod_{l=1}^{N-1} a_l^{s_l} \nonumber \\
 + (-1)^{s_N} &\sum_{\vec{a}^{N-1} \in \{-1,1\}^{\otimes N-1} }\sqrt{2}\cos\left(\frac{\pi}{4}(N+2-a')\right)\prod_{l=1}^{N-1} a_l^{s_l}	\,.
\end{align}
As $\cos(x+\frac{\pi}{2})=-\sin(x)$, we have
\begin{align}
\beta\left(\vec{s}^N\right) &= \beta\left(\vec{s}^{N-1}\right) - (-1)^{s_N} \gamma\left(\vec{s}^{N-1}\right)	
\,.\label{eq:betarecursion}
\end{align}
Using a similar argument, we have
\begin{align}
 \gamma\left(\vec{s}^{N-1},1-s_N\right) &= \gamma\left(\vec{s}^{N-1}\right) + (-1)^{1-s_N}
 \beta\left(\vec{s}^{N-1}\right)	\nonumber \\
 &=-(-1)^{s_N}\beta\left(\vec{s}^N\right)\,, \label{eq:gammarecursion1}
\end{align}
or equivalently,
\begin{equation}
 \gamma\left(\vec{s}^{N}\right) =-(-1)^{1-s_N}\beta\left(\vec{s}^{N-1},1-s_N\right)\,. \label{eq:gammarecursion2}
\end{equation}
By the induction hypothesis, let us suppose that Eq.~\eqref{eq:equivalentbetas} is true for some $N=n_0-1$, where $n_0\ge3$, we will now prove that it also holds for $N=n_0$. Explicitly, note from Eq.~\eqref{eq:betarecursion} that
\begin{align*}
\beta\left(\vec{s}\,^\nn\right) &= \beta\left(\vec{s}\,^{\nn-1}\right) - (-1)^{s_\nn} \gamma\left(\vec{s}\,^{\nn-1}\right)\\
&= \beta\left(\vec{s}\,^{\nn-1}\right) - (-1)^{s_\nn+s_{\nn-1}} \beta\left(\vec{s}\,^{\nn-2},1-s_{\nn-1}\right)\\
&= 2^{\frac{\nn}{2}}\cos\left[\frac{\pi}{4}(\nn-2s')\right]- (-1)^{s_\nn+s_{\nn-1}}2^{\frac{\nn}{2}}\times\\
&~\quad\cos\left[\frac{\pi}{4}(\nn-2(s'+1-2s_{\nn-1}))\right]\\
&= 2^{\frac{\nn}{2}}\bigg\{\cos\left[\frac{\pi}{4}(\nn-2s')\right]- (-1)^{s_\nn}\times\\
&\qquad\qquad\cos\left[\frac{\pi}{4}(\nn-2s'-2)\right]\bigg\}\\
&= 2^{\frac{\nn}{2}}\bigg\{\cos\left[\frac{\pi}{4}(\nn-2s')\right]+\\
&\qquad\qquad\cos\left[\frac{\pi}{4}(\nn-2s'+2-4s_\nn)\right]\bigg\}\\
&= 2^{\frac{\nn}{2}+1}\cos\left[\frac{\pi}{4}(\nn+1-2s)\right]\cos\left[\frac{\pi}{4}(2s_\nn -1)\right]\\
&= 2^{\frac{\nn+1}{2}}\cos\left[\frac{\pi}{4}(\nn+1-2s)\right]
\end{align*}
where the second equality follows from Eq.~\eqref{eq:gammarecursion2}, the third equality follows from the induction hypothesis, and the other equalities follow from simple algebraic calculation using trigonometric identities.
\end{proof}

\section{Equivalent MABK Inequalities}\label{App:EquivalentMABK}

Bell inequalities in the same equivalence class are those that can be obtained from one another by some permutation of the
labels on the parties $k$, and/or settings $s_k=0\leftrightarrow s_k=1$ and/or outcomes ``+1" $\leftrightarrow$ ``-1" in the coefficients defining the inequality~\cite{masanes2003,collins2004} (see also Ref.~\cite{WWproveWWZB}). Testing a given set of correlation functions against a different, but equivalent Bell inequality amounts to testing the \textit{same} Bell inequality, but with a different labeling and/or sign on the correlation functions.

Given that there are $N!$ permutations on the label $k$, $2^N$ distinct permutations on the labels $s_k$, and two different labeling of outcomes for each of the $2N$ measurement directions, the number of inequalities that are equivalent to Eq.~\eqref{eq:betanew} is upper bounded by $N!2^{3N}$. However, as we will show, most of these relabelings give identical inequalities. To this end, let us start by proving the following lemma.

\begin{lem}\label{Lem:PermOutcome}
For the MABK inequality in Eq.~\eqref{Eq:MABK:S1}, any relabeling of the measurement outcomes on a subset of parties $k\in\N$ can be simulated by a permutation of the label $s_k$ for the remaining parties.
\end{lem}

\begin{proof}
Firstly, note that the effect of relabeling the measurement outcomes for the $j$-th measurement setting of the $k^{th}$ party modifies the correlation functions $E(\vec{s})$ by a phase factor, i.e.,
\begin{equation}
 E(\vec{s})\to (-1)^{1+j-s_k}E(\vec{s})\,.
\end{equation}
From this, we can see that relabeling the outcomes of both $s_k=0$ and 1 will only introduce a global sign change, which has no effect because of the absolute value function. Henceforth, we therefore only consider the case where the outcomes of the measurement corresponding to $s_k = 1$ are relabelled.

For simplicity, let us also consider the scenario where only the label for the measurement outcome of the $k^{th}$ party is changed (the proof for the more general scenario proceeds analogously). We can absorb the effect of this relabeling into the $\beta$'s by setting
\begin{equation}
\beta'(s,N) = (-1)^{s_k}\beta(s,N) \,.\label{eq:outcomerelabelled}
\end{equation}
We now show that this sign change can be simulated by a change of the label $s_k$ for the remaining $N-1$ observers, i.e., by setting
\begin{align}
s'_l = \delta_{kl}\,s_l +(1-\delta_{kl})(1 - s_l)
\end{align}
With this new labeling, we have
\begin{equation}
 s' = \sum_{l=1}^N s'_l = 2s_k - 1 + \sum_{l=1}^N (1-s_l)
 = N - 1 - s + 2s_k	\,.
\end{equation}
Substituting this into Eq.~\eqref{eq:betanew} gives
\begin{align}
\beta\left(s',N\right) &= 2^{\frac{N+1}{2}}\cos\left(\frac{\pi}{4}(-N + 3 + 2s - 4s_k))\right) \nonumber \\
&= -(-1)^{s_k} 2^{\frac{N+1}{2}}\cos\left(\frac{\pi}{4}(-N -1 + 2s)\right) \nonumber \\
&= -(-1)^{s_k} 2^{\frac{N+1}{2}}\cos\left(\frac{\pi}{4}(N + 1 - 2s)\right) \nonumber \\
&= -(-1)^{s_k}\beta(s,N) \,,
\end{align}
which is the same as Eq.~\eqref{eq:outcomerelabelled}, except for an overall sign which has no effect because of the absolute value function.
\end{proof}

It is then a small step away to prove the following Theorem.

\begin{thm}\label{Thm:MABK:Equivalence}
There are $2^{N}$ Bell inequalities that are equivalent to Eq.~\eqref{Eq:MABK:S1} under relabelings of the measurement outcomes, measurement labels or permutations of observers.
\end{thm}

\begin{proof}
There are $2^N$ independent permutations on the measurement settings, mapping $s_k=0\leftrightarrow s_k=1$, which give different inequalities since all the $s_k$'s are independent. By Lemma~\ref{Lem:PermOutcome}, we know that the permutation of measurement outcomes do not introduce any new inequality beyond this set of $2^N$ inequalities.

Next, note that a permutation of observers rearranges the label $k$, which corresponds to a permutation of the entries in the vector $\vec{s}=\left(s_1,\ldots,s_N\right)$. Obviously, this leaves the sum of the entries in $\vec{s}$ unchanged. Because $\beta\left(s,N\right)$ only depends on $s$ and $N$, $\beta$ is invariant under permutations of the entries in $\vec{s}$, so permutations of the observers simply rearrange the terms in Eq.~\eqref{Eq:MABK:S1}.
\end{proof}

%
%
%\bibitem{Horodecki:RMP:entanglement} R.~Horodecki \textit{et al.}, \rmp {\bf
% 81}, 865 (2009).
%
%\bibitem{P.Lougovski:PRA:034302} P.~Lougovski and S.~J.~van~Enk, Phys.~Rev.~A {\bf
% 80}, 034302 (2009).
%
%\bibitem{FullEntanglement} D.~Collins {\em et al.}, Phys.~Rev.~Lett. {\bf 88}, 170405
% (2002); M.~Seevinck and G.~Svetlichny, {\em ibid.} {\bf 89}, 060401 (2002).
%
%\bibitem{S.Boyd:Book:2004} S.~Boyd and L.~Vandenberghe, {\em Convex
% Optimization} (Cambridge, New York, 2004).

\end{document}